\documentclass[10pt]{article}

\usepackage{wrapfig}
\usepackage{sectsty}
\sectionfont{\large}
\subsectionfont{\normalsize}
\subsubsectionfont{\small}
\usepackage{listings}
\usepackage{caption}
\usepackage{subcaption}

\usepackage[title]{appendix}

\usepackage{amsfonts}
\usepackage{microtype}
\usepackage{graphicx}
\usepackage{booktabs}
\usepackage{multicol}
\usepackage{algorithm}
\usepackage{algpseudocode}
\usepackage{enumitem}

\usepackage{setspace}
\usepackage[longnamesfirst]{natbib}

\usepackage{amsmath}
\usepackage{amssymb}
\usepackage{mathtools}
\usepackage{amsthm}

\usepackage{amsfonts}
\usepackage{amsmath}
\makeatletter
\newcommand*{\da@rightarrow}{\mathchar"0\hexnumber@\symAMSa 4B}
\newcommand*{\da@leftarrow}{\mathchar"0\hexnumber@\symAMSa 4C}
\newcommand*{\xdashrightarrow}[2][]{%
  \mathrel{%
    \mathpalette{\da@xarrow{#1}{#2}{}\da@rightarrow{\,}{}}{}%
 }%
}
\newcommand{\xdashleftarrow}[2][]{%
  \mathrel{%
    \mathpalette{\da@xarrow{#1}{#2}\da@leftarrow{}{}{\,}}{}%
 }%
}
\newcommand*{\da@xarrow}[7]{%
  \sbox0{$\ifx#7\scriptstyle\scriptscriptstyle\else\scriptstyle\fi#5#1#6\m@th$}%
  \sbox2{$\ifx#7\scriptstyle\scriptscriptstyle\else\scriptstyle\fi#5#2#6\m@th$}%
  \sbox4{$#7\dabar@\m@th$}%
  \dimen@=\wd0 %
  \ifdim\wd2 >\dimen@
    \dimen@=\wd2 %
  \fi
  \count@=2 %
  \def\da@bars{\dabar@\dabar@}%
  \@whiledim\count@\wd4<\dimen@\do{%
    \advance\count@\@ne
    \expandafter\def\expandafter\da@bars\expandafter{%
      \da@bars
      \dabar@ 
   }%
 }%
  \mathrel{#3}%
  \mathrel{%
    \mathop{\da@bars}\limits
    \ifx\\#1\\%
    \else
      _{\copy0}%
    \fi
    \ifx\\#2\\%
    \else
      ^{\copy2}%
    \fi
 }%
  \mathrel{#4}%
}
\makeatother

\usepackage[colorlinks,
            linkcolor=red,
            anchorcolor=blue,
            citecolor=blue
            ]{hyperref}
\usepackage[capitalize,noabbrev]{cleveref}

\theoremstyle{plain}
\newtheorem{theorem}{Theorem}[section]
\newtheorem{proposition}[theorem]{Proposition}
\newtheorem{lemma}[theorem]{Lemma}
\newtheorem{corollary}[theorem]{Corollary}

\theoremstyle{definition}
\newtheorem{definition}[theorem]{Definition}
\newtheorem{assumption}[theorem]{Assumption}

\newtheorem{remark}[theorem]{Remark}

\usepackage{smile}
\usepackage{makecell}

\usepackage{amsmath,amsfonts,bm}

\def\1{\bm{1}}

\def\vzero{{\bm{0}}}
\def\vone{{\bm{1}}}

\def\valpha{{\bm{\alpha}}}
\def\vbeta{{\bm{\beta}}}
\def\vepsilon{{\bm{\epsilon}}}

\def\va{{\bm{a}}}
\def\vb{{\bm{b}}}
\def\vc{{\bm{c}}}
\def\vd{{\bm{d}}}
\def\ve{{\bm{e}}}
\def\vf{{\bm{f}}}

\def\vp{{\bm{p}}}
\def\vq{{\bm{q}}}

\def\vs{{\bm{s}}}

\def\vu{{\bm{u}}}
\def\vv{{\bm{v}}}

\def\vx{{\bm{x}}}
\def\vy{{\bm{y}}}

\def\eva{{a}}

\def\evd{{d}}

\def\evf{{f}}

\def\evh{{h}}

\def\evp{{p}}

\def\evs{{s}}

\def\evu{{u}}
\def\evv{{v}}
\def\evw{{w}}
\def\evx{{x}}

\def\mA{{\bm{A}}}

\def\mLambda{{\bm{\Lambda}}}
\def\mSigma{{\bm{\Sigma}}}

\DeclareMathAlphabet{\mathsfit}{\encodingdefault}{\sfdefault}{m}{sl}
\SetMathAlphabet{\mathsfit}{bold}{\encodingdefault}{\sfdefault}{bx}{n}

\def\gG{{\mathcal{G}}}

\def\sA{{\mathbb{A}}}

\def\sK{{\mathbb{K}}}

\def\sN{{\mathbb{N}}}

\def\sP{{\mathbb{P}}}

\def\sR{{\mathbb{R}}}

\def\sW{{\mathbb{W}}}
\def\sX{{\mathbb{X}}}

\def\emLambda{{\Lambda}}

\def\emSigma{{\Sigma}}

\DeclareMathOperator*{\argmin}{arg\,min}

\DeclareMathOperator{\im}{im}

\newcommand*{\T}{\mathsf{T}}
\DeclareMathOperator\supp{supp}

\DeclareMathOperator{\spn}{span}

\newcommand{\interior}[1]{%
  {\kern0pt#1}^{\mathrm{o}}%
}

\newcommand\extrafootertext[1]{%
    \bgroup
    \renewcommand\thefootnote{\fnsymbol{footnote}}%
    \renewcommand\thempfootnote{\fnsymbol{mpfootnote}}%
    \footnotetext[0]{#1}%
    \egroup
}
\usepackage[textsize=tiny]{todonotes}
\usepackage{fullpage}

\usepackage{float}
\labelwidth\leftmargini\advance\labelwidth-\labelsep \labelsep 5pt

\def\@listi{\leftmargin\leftmargini}
\def\@listii{\leftmargin\leftmarginii
   \labelwidth\leftmarginii\advance\labelwidth-\labelsep
   \topsep 2pt plus 1pt minus 0.5pt
   \parsep 1pt plus 0.5pt minus 0.5pt
   \itemsep \parsep}
\def\@listiii{\leftmargin\leftmarginiii
    \labelwidth\leftmarginiii\advance\labelwidth-\labelsep
    \topsep 1pt plus 0.5pt minus 0.5pt
    \parsep \z@ \partopsep 0.5pt plus 0pt minus 0.5pt
    \itemsep \topsep}
\def\@listiv{\leftmargin\leftmarginiv
     \labelwidth\leftmarginiv\advance\labelwidth-\labelsep}
\def\@listv{\leftmargin\leftmarginv
     \labelwidth\leftmarginv\advance\labelwidth-\labelsep}
\def\@listvi{\leftmargin\leftmarginvi
     \labelwidth\leftmarginvi\advance\labelwidth-\labelsep}
     
\usepackage{authblk}

\usepackage{titlesec}

\titlespacing\section{0pt}{12pt plus 4pt minus 2pt}{0pt plus 2pt minus 2pt}
\titlespacing\subsection{0pt}{12pt plus 4pt minus 2pt}{0pt plus 2pt minus 2pt}
\titlespacing\subsubsection{0pt}{12pt plus 4pt minus 2pt}{0pt plus 2pt minus 2pt}

\title{\textbf{\Large{
A Day-to-Day Dynamical Approach to the Most Likely \\User Equilibrium  Problem\\[10pt]}}}

\author[1]{\normalsize Jiayang Li}
\author[1]{\normalsize Qianni Wang}
\author[2]{\normalsize Liyang Feng}

\author[2]{\normalsize Jun Xie}

\author[1]{\normalsize Yu (Marco) Nie\footnote{Corresponding author; \texttt{y-nie@northwestern.edu}.}}

\affil[1]{\small \textit{Department of Civil and Environmental Engineering, Northwestern University, USA.}}

\affil[2]{\small \textit{Department of School of Transportation and Logistics, Southwest Jiaotong University, China.}}

\date{}

\begin{document}

\maketitle

\begin{abstract}
	\noindent The lack of a unique user equilibrium (UE) route flow in traffic assignment has posed a significant challenge to many transportation applications. The maximum-entropy principle, which advocates for the consistent selection of the most likely solution as a representative, is often used to address the challenge. Built on a recently proposed day-to-day (DTD) {discrete-time} dynamical model called cumulative logit (CULO), this study provides a new behavioral underpinning for the maximum-entropy UE (MEUE) route flow.  It has been proven that CULO can reach a UE state without presuming travelers are perfectly rational. Here, we further establish that CULO always converges to the MEUE route flow if (i) travelers have zero prior information about routes and thus are forced to give all routes an equal choice probability, or (ii)  all travelers gather information from the same source such that the so-called general proportionality condition is satisfied.  Thus, CULO may be used as a practical solution algorithm for the MEUE problem. To put this idea into practice, we propose to eliminate the route enumeration requirement of the original CULO model through an iterative route discovery scheme.  We also examine {the discrete-time versions of four popular continuous-time dynamical models} and compare them to CULO.  The analysis shows that the replicator dynamic is the only one that has the potential to reach the MEUE solution with some regularity.  The analytical results are confirmed through numerical experiments.
\end{abstract}

{\it Keywords:}  maximum entropy; traffic assignment; cumulative logit; day-to-day dynamical model; general proportionality condition

\section{Introduction}

A fundamental problem in transportation systems analysis is predicting the distribution of traffic over routes connecting each origin-destination (OD) pair in a general congestible network, commonly known as traffic assignment \citep{beckmann1956studies}.  In transportation planning, traffic assignment is often framed as a non-cooperative routing game in which travelers' selfish route choices drive network-wide traffic distribution toward a user equilibrium (UE) state \citep{wardrop1952road,roughgarden2002bad}. Generally speaking, neither the set of routes used at UE nor the number of travelers selecting these routes (called route flow) is unique \citep{sheffi1985urban}. Indeed, there are potentially infinitely many route flows that correspond to a UE state of the routing game. 
Practitioners used to accept any UE route flow that emerges from a traffic assignment procedure, electing to ignore this nuance altogether. However, even when aggregate assignment results (e.g., the link flow) are not affected by the lack of uniqueness,  this practice may undermine any applications that depend on UE route flows \citep[e.g., select link analysis, see][]{bar2012user}. The problem is that using an arbitrary UE route flow can be difficult to justify, and, more importantly, such a flow may vary disproportionately with small perturbations in system inputs \citep{lu2010stability}. {The same problem also affects multi-class traffic assignment models, in which travelers are classified into groups based on their individual characteristics, such as the value of time.  In these models, the class-specific link flow, in addition to the route flow, is often non-unique at UE \citep{bar2012user}. This constitutes a serious concern for any efforts to understand the distributional effects of certain control and/or management policies, such as equity analysis \citep{wang2023entropy} or mixed-autonomy traffic analysis \citep{bahrami2020optimal}.}

It has been suggested that additional criteria may be imposed to rank the UE route flows, and a decision-maker should stick to the highest-ranked flow to maintain consistency and stability of the decision process.   \citet{lu2010stability} showed such a rank could be produced by maximizing a suitable function of UE route flows.  Yet, this does not solve the issue of justification since ``suitable" objective functions are countless, and there seems hardly any good reason to prefer one to another.   The only exception, to the best of our knowledge, is the entropy function  \citep{rossi1989entropy,akamatsu1997decomposition,bell1997transportation,bar1999route,larsson2001}.  Selecting the UE route flow that maximizes entropy is justified by the fact that such flow is the \emph{most likely} to be observed given the prior information, i.e., adherence to UE  by travelers. This principle, widely used in statistical mechanics and information theory, can also be interpreted as a claim of maximum ignorance beyond what is firmly known {by the modeler}. 

Despite its popularity, the maximum-entropy UE (MEUE) route flow lacks a solid micro-behavioral foundation.  It remains an open question what, if any, route choice behaviors can consistently lead the routing game to such a flow.  \citet{bar1999route} noted an MEUE route flow always distributes traffic to two paired equal-cost alternative segments by the same proportion regardless of travelers' origin or destination.  This observation connects MEUE route flows to route choice behaviors and, in so doing, provides a scalable solution method for the MEUE route flow problem \citep{bar2006primal,bar2010traffic,xie2019new}. 
Using a large taxi trajectory data set, \citet{xie2017testing} showed that \emph{proportionality}, as it is often referred to in the literature, is approximately satisfied among taxi trips.  However, proportionality is an aggregate result of route choice that cannot be easily linked to individual behaviors. It is one thing to observe travelers obey the condition of proportionality \emph{collectively}, but quite another to explain why they behave this way \emph{individually}.  Moreover, proportionality between paired alternative segments is a necessary but insufficient condition for entropy maximization \citep{bar2006primal}.  Sufficiency requires high-order proportionality conditions \citep{borchers2015traffic}, but enforcing them weakens not only the behavioral interpretation of proportionality but also the scalability of the solution methods derived from it.

MEUE may also be viewed as a limit of the \emph{stochastic user equilibrium} (SUE).  SUE is a ``perturbed" UE where travelers, subject to perception errors, elect to choose the route ``believed" to be the best \citep{daganzo1977stochastic} through a random utility model \citep{ben1985discrete}. A well-known result in transportation is that SUE  approaches UE when perception errors are reduced to zero \citep{fisk1980some}. In game theory, this is known as the purification theorem \citep{harsanyi1973games}. Moreover, if travelers' choices are given by the logit model \citep{mcfadden1973conditional}, the limiting --- or ``purified"  --- SUE would coincide with MEUE \citep{larsson2001, mamun2011select}.
However, interpreting MEUE as a limit of SUE implies it could be reached only if travelers always select the best route --- an assumption widely contested in the literature \citep[see, e.g.,][]{simon1955behavioral}. Moreover,  that SUE \emph{can} be steered toward MEUE by tweaking its parameters does not mean travelers are likely to behave accordingly. Indeed, it is unclear whether, why, and how the perception errors should gradually decrease to zero from a behavioral point of view.

In this paper, we attempt to 
provide a new behavioral foundation for the MEUE route flow using a day-to-day (DTD) dynamical approach.  In part, our effort is inspired by a recently developed DTD dynamical model called CULO \citep{li2023wardrop}, which is capable of reaching a UE state of the routing game under the presumption that travelers are \textit{not} perfectly rational even at the equilibrium. CULO describes how travelers gradually adjust their route valuations, hence choice probabilities, based on past experiences.  A crucial difference between CULO and the classical DTD models \citep[e.g.,][]{horowitz1984stability, cascetta1993modelling, watling1999stability, watling2003dynamics}
is route valuation: whereas classical models value routes based on the cost averaged over time,  CULO values them based on the \emph{cumulative} cost. As a result, CULO converges to UE globally under mild conditions, while other similar DTD models converge to SUE \citep{horowitz1984stability, cascetta1993modelling, watling1999stability}. In numerical experiments, \cite{li2023wardrop} discovered that CULO can converge to the MEUE route flow when starting from a certain initial point, notably an equal-distribution route flow (obtained by assigning the same choice probability to all routes between the same OD pair). This finding is intriguing because it indicates MEUE may be obtained from a simple and behaviorally sound DTD process, a possibility that, to the best of our knowledge, has never been discussed in the literature before. Once confirmed, it would not only help explain how the MEUE route flow may emerge from the evolution of imperfect route choices but also give a general algorithm for finding such a flow.  Motivated by this observation, we set out in this study to \emph{identify the conditions under which the convergence of CULO to MEUE is guaranteed}.

Originally, CULO assumes travelers actively consider all routes or at least a  set that covers all UE routes at the beginning.  In reality, such a route set may be either unknown to the travelers prior or simply too large to be included in the decision process. \citet{xie2019new} discovered a case in which the number of UE routes for a single OD pair can be as many as more than half a billion. There are also considerable cross-OD variations. For example, \citet{bar2005user} noted up to 2000 routes could be used at UE for some OD pairs in the Chicago regional network, though travelers from most OD pairs settle for one to two UE routes.  Hence, we further propose to iteratively generate the route set in CULO,  assuming travelers continuously explore the vast route space and attempt to strike a balance between exploration, i.e., discovering new routes, and exploitation, i.e., making the best use of the routes found so far. This concept of exploration vs. exploitation is central to bandit problems and reinforcement learning problems \citep{bush1955stochastic}. It also bears similarities with the use of column generation --- which generates routes on the fly --- in traffic assignment \citep{jayakrishnan1994faster}. Can the convergence of CULO toward MEUE still be secured with route discovery? That is the second question to be explored in our study. 

CULO is unique in the literature not because it converges to UE globally but because it does so by allowing explicit learning and deviation from perfect rationality. Many other dynamical models --- the vast majority of which are continuous-time models --- are known to converge to UE. For instance, the Smith dynamic \citep{smith1984stability} moves flow between every pair of routes at a rate proportional to the product of the flow on the higher-cost route and the cost difference.  The projection dynamic is a continuous-time version of the projection method for solving variational inequality problems \citep{dupuis1993dynamical, friesz1994day, zhang1996local}. Some evolutionary dynamics from game theory \citep{weibull1997evolutionary, sandholm2010population} have also been adapted to study routing games \citep[see, e.g.,][]{yang2009day, li2022differentiable}. What is the relationship between MEUE and the equilibrium solutions achieved by these models? That is our third question.

\subsection{Our contributions}

Our first and foremost result is that the limiting point of CULO  minimizes the ``distance" from the initial solution (corresponding to travelers' initial route valuation) to the set of UE route flows (referred to as the UE set hereafter), as measured by the Kullback–Leibler (KL) divergence. In other words, running CULO until convergence is equivalent to ``KL projecting" the initial solution onto the UE set. This result is then used to establish several useful properties for CULO. First, if CULO does converge, it always admits the same UE route flow from each initial solution.  This property ensures the behavioral parameters in CULO, which may affect the trajectory of convergence, do not affect the equilibrium state.  Second, the limiting point of CULO changes continuously with the initial solution,  which prevents the dynamical model from suffering large prediction errors caused by inaccurate information about the initial state.  Third,  all routes that \emph{may} be used by a UE route flow --- called the UE routes hereafter --- \emph{will} be used at the limiting point of CULO, provided they are included in the choice set from the beginning. 
Combining the first two properties above gives us the EUC (existence, uniqueness, and continuity of solutions) condition described in \citet{sandholm2005excess}, which is part of the ``desiderata" for an ideal dynamical model. The third one is a necessary condition for achieving MEUE, sometimes known as ``no-route-left-behind" policy \citep{bar1999route}. 

We also identify and verify the conditions that can steer CULO to MEUE based on the above results. We confirm that starting from the equal-distribution route flow is indeed one of them. Intuitively, this does make sense: if no one has prior information about the routes, then equal distribution is the logical and entropy-maximizing outcome. CULO simply preserves this property throughout the KL projection process.  Yet, we also show equal distribution is but one of infinitely many MEUE-inducing initial solutions. One general requirement is that the initial valuation on any route equal the sum of the valuations on the links used by the route, and the link valuations are identical for all routes.

Our third result concerns how to enhance CULO with a route discovery module.   Integrating route discovery with CULO requires strategies to (i) initialize valuation on newly found routes and (ii) encourage travelers to explore routes beyond the best ones.  For (i), we propose to keep a vector of cumulative link valuations from which the cumulative valuation on any route can be obtained without knowing the details about the evolution history. To enhance exploration, white noise is added to link valuations whenever travelers attempt to search for new routes, which allows them to explore a greater portion of the route space and, consequently, to come across and retain more non-UE paths in the choice set.  Such redundancy is necessary to ensure no path is left behind.  As a by-product, CULO is turned from an instrument for analysis into a practical solution algorithm for the MEUE route flow problem. Unlike most algorithms proposed for this problem \citep[e.g.,][]{bell1997transportation,larsson2001,bar2006primal, xie2019new}, the CULO algorithm does not view it as a constrained optimization problem. Instead, it simply mimics the evolutionary process by which the routing game converges. CULO may not be as efficient --- in terms of both memory consumption and computation time --- as the state-of-the-art algorithms such as the bush-based algorithm of \citet{xie2019new}, but it compensates for this shortcoming with simplicity and robustness. Indeed, implementing CULO requires little more than a standard shortest path algorithm plus the ability (and computer memory) to manage routes found in the dynamical process. It is also designed to find the exact MEUE solution rather than an approximation that may fail to satisfy higher-order proportionality conditions. 
Thus, for small to medium applications that need a high-quality MEUE route flow, CULO offers a rather appealing alternative.

{Last but not least, we examine a group of DTD dynamical models that are known to converge to UE, while focusing on their ability to reach MEUE under similar conditions. 
    Well known in their continuous-time form,  these models are discretized in this study to strengthen the behavioral representation, i.e., to reflect the fact that route choice is not continuously adjustable in time  \citep{watling2003dynamics}.
}
Although only numerical findings are available due to analytical difficulties, the insights are new and interesting.   We shall see that the popular replicator dynamic \citep{taylor1978evolutionary} demonstrates a surprisingly strong potential to find a near-MEUE solution.  Its performance tracks that of CULO closely, despite the fact they are completely different models in appearance. On the other hand, neither Smith's \citep{smith1984stability} nor the projection \citep{friesz1994day, zhang1996local} or the best-response dynamic \citep{gilboa1991social} is capable of getting close to MEUE.   All violate the ``no-route-left-behind" policy in our experiments.

\subsection{Organization}

The rest of the paper is organized as follows. Section \ref{sec:setting} sets up the problem and discusses related works. In Section \ref{sec:main}, we prove our main result, which establishes that running CULO until it converges is equivalent to performing a KL projection of the initial route choice onto the set of UE. Building on this foundation, we then conduct an analysis of CULO and identify specific conditions that lead to its convergence at MEUE.  Section \ref{sec:generation} addresses the issue of route space exploration and {Section \ref{sec:other} examines and compares the discretized version of several continuous-time dynamical models with CULO.}
Results of numerical experiments designed to validate the analyses are reported in  Section \ref{sec:experiments}.  Section \ref{sec:conclusion} concludes the paper.

\subsection{Notation}

We use $\sR$ and $\sR_+$ to denote, respectively, the set of real numbers and non-negative real numbers,  and use $\bar \sR = \sR \cup \{\infty, -\infty\}$ to denote the set of extended real numbers. 
For a vector $\va \in \sR^n$, we denote $\|a\|_p$ as its $\ell_p$ norm and denote { $\supp{(\va)} = \{i \in [n]: \eva_i > 0\}$ ($[n] = \{1, \ldots, n \}$)} as its support and $\diag(\va)$ as a square diagonal matrix with the elements of vector $\va$ on the main diagonal. 
For a matrix $\mA \in \sR^{n \times m}$, we denote $\|\mA\|_p$ as its matrix norm induced by the vector $\ell_p$ norm, denote $\ker(\mA) = \{\vx \in \sR^m: \mA \vx = 0\}$ as its kernel, and denote $\im(\mA) = \{\vy \in \sR^m: \vy = \mA \vx, \ \vx \in \sR^{n}\}$ as its image. For two vectors $\va, \vb \in \sR^n$, their inner product is denoted as $\langle \va, \vb \rangle$. For a finite set $\sA$, we write $|\sA|$ as the number of elements in $\sA$ and $2^{\sA}$ as the set of all subsets of $\sA$.  For a real number $a \in \sA$, we denote $[a]_+ = \max\{a, 0\}$. Given a set of vectors $\va_1, \ldots, \va_n \in \sR^m$, we denote their linear span as $\spn{(\va_1, \ldots, \va_n)} = \{\sum_{i = 1}^n \lambda 
_i \cdot \va_i: \lambda_i \in \sR, \ i = 1, \ldots, n\}$. Given any set $\sA \subseteq \sR^m$, we define its orthogonal complement as $\sA^{\perp} = \{\vx \in \sR^m: \langle \vx, \vy \rangle = 0 , \ \forall \vy \in \sA\}$.

\section{Problem setting and preliminaries}
\label{sec:setting}

We model a transportation network as a directed graph $\gG(\sN, \sA)$, where $\sN$ and $\sA$ are the set of nodes and links, respectively. Let $\sW \subseteq \sN \times \sN$ be the set of OD pairs and $\sK \subseteq 2^{\sA}$ be the set of available routes connecting all OD pairs. We use $\sK_w \subseteq \sK$ to denote the set of routes connecting $w\in\sW$ and $\sA_k \subseteq \sA$ the set of all links on route $k \in \sK$. Also, denote $\emSigma_{w,k}$ as the OD-route incidence with $\emSigma_{w,k} = 1$ if the route $k \in \sK_w$ and 0 otherwise; and $\emLambda_{e,k}$ as the link-route incidence, with $\emLambda_{e,k} = 1$ if $e \in \sA_k$ and 0 otherwise. We write $\mLambda = (\emLambda_{e,k})_{e \in \sA, k \in \sK}$ and $\mSigma = (\emSigma_{w,k})_{w \in \sW, k \in \sK}$. Let $\vd = (\evd_w)_{w \in \sW}$ be a vector with $\evd_w$ denoting the number of travelers between $w \in \sW$.  All travelers are identical, and their route choice strategy is represented by a vector $\vp = (\evp_k)_{k \in \sK}$, where $\evp_k$ is the \textit{proportion of travelers} selecting $k\in \sK_w$. The feasible set for $\vp$ can be written as
$\sP = \{\vp \in \sR_+^{|\sK|}: \mSigma \vp = \vone\}$. Let $\vf = (\evf_k)_{k \in \sK}$ and $\vx = (\evx_a)_{a \in \sA}$, with $\evf_k$ and $\evx_a$ being the flow (i.e., number of travelers) on route $k$ and link $a$, respectively. It follows $\vf = \diag(\vq) \vp$ (where $\vq = \mSigma^{\T} \vd$) and $\mLambda \vf = \vx$.
Further define $\vu = (\evu_a)_{a \in \sA}$ as a vector of link cost, determined by a  function $u(\vx) = (u_a(\vx))_{a \in \sA}$. Then, the vector of route cost  $\vc = \mLambda^{\T} \vu$. To summarize, the route cost function $c: \sP \to \sR^{|\sK|}$ can be defined as
$
    c(\vp) = \mLambda^{\T} \vu = \mLambda^{\T} u(\mLambda \vf) = \mLambda^{\T} u(\mLambda \diag(\vq) \vp).
$
For notational simplicity, we also introduce the symbol $\bar \mLambda = \mLambda \diag(\vq)$  so that $\vx$ can be written as $\bar\mLambda \vp$.

Throughout the paper, we impose two assumptions on the link cost function $u(\vx)$, {whose domain (the set of feasible link flows) is written as $\sX = \{\vx: \sR^{|\sA|}: \vx = \bar \mLambda \vp, \ \vp \in \sP\}$.}

\begin{assumption}
\label{ass:differentiable}
    The link cost function $u(\vx)$ is continuously differentiable and non-negative on $\sX$. 
\end{assumption}
{
\begin{assumption}
\label{ass:monotone}
    The link cost function $u(\vx)$ is strictly monotone on $\sX$, i.e., i.e., $\langle u(\vx) - u(\vx'), \vx - \vx' \rangle > 0$ for all $\vx, \vx' \in \sX$. 
\end{assumption}
}

Travelers are viewed as playing a routing game by choosing a mixed strategy $\vp$ to minimize their own travel costs. 
Those from the same OD pair adopt the same mixed strategy, and per the law of large numbers, $\vp$ gives the proportion of the travelers from each OD pair selecting each route connecting that OD pair.
We define a user equilibrium (UE) route choice strategy of the routing game \citep{wardrop1952road} as follows.
\begin{definition}[UE strategy]
A route choice strategy $\vp^* \in \sP$ is a user equilibrium strategy  if
$c_k(\vp^*) > \min_{k' \in \sK_w} c_{k'}(\vp^*)$ implies $\evp_k^* = 0$ for all $w \in \sW$ and $k \in \sK_w$.
\end{definition}

\begin{proposition}[\citet{dafermos1980traffic}]
\label{prop:ue-vi}
    A route choice strategy $\vp^*$ is a UE strategy if and only if it solves the following  variational inequality (VI) problem: find $\vp^* \in \sP$ such that
\begin{equation}
        \langle c(\vp^*), \vp - \vp^* \rangle \geq \vzero, \quad \forall \vp \in \sP.
        \label{eq:ue-vi}
    \end{equation}
\end{proposition}

Denoting the solution set to the above VI problem as $\sP^*$, the following two propositions, both established by \citet{dafermos1980traffic}, characterize the geometry of $\sP^*$. 
\begin{proposition}
\label{prop:solution-set-2}
If $c(\vp)$ is {strictly monotone on $\sP$}, then $\sP^*$ is a singleton.
\end{proposition}
\begin{proposition}
\label{prop:solution-set}
If $u(\vx)$ is  {strictly monotone on $\sX$}, then $\sX^* = \{\vx^* = \bar \mLambda \vp^*: \vp^* \in \sP^*\}$ is a singleton. Moreover, $\sP^*$ can be represented as a polyhedron $\{\vp^* \in \sP: \bar \mLambda \vp^* = \vx^*\}$, where $\vx^*$ is the unique UE link flow.
\end{proposition}
When the function $u(\vx)$ is {strictly} monotone, the {strict} monotonicity of $c(\vp)$ can be guaranteed if $\mLambda$ has a full column rank. This condition, however, is rarely satisfied in the networks of practical interest. Hence, the UE strategy $\vp^*$ (hence the UE route flow $\vf^*$) is usually not unique.

In what follows, Section \ref{sec:meue} introduces the MEUE problem, including the formulation, basic properties,   and the relationship with the logit-based stochastic user equilibrium (SUE) model. In Section \ref{sec:culo}, we present the CULO model developed in \cite{li2023wardrop} and contrast it with the classical DTD model \citep{horowitz1984stability}.

\subsection{The MEUE problem}
\label{sec:meue}

To consistently select a unique UE strategy from $\sP$, one may define another function of $p\in \sP$ that admits a unique extreme value \citep{lu2010stability}. The most widely used function is the \textit{negative} entropy function.
 \citet{rossi1989entropy} defined the negative entropy of any $\vp \in \sP$ as
\begin{equation}
    \phi(\vp) = \langle \diag(\vq) \vp, \log(\vp)\rangle,
\end{equation}
which measures the number of different ways travelers can be arranged to produce the route flow corresponding to $\vp$ {(see Appendix \ref{app:entropy} for a detailed explanation)}. The lower the value of $\phi(\vp)$, the more likely to occur the route flow associated with $\vp$. Thus, maximizing entropy, or minimizing $\phi(\vp)$,  is expected to produce the most likely outcome.  

\begin{definition}[Maximum-entropy user equilibrium, or MEUE]
A route choice strategy $\bar \vp^* \in \sP$ corresponds to the MEUE route flow or the most likely route flow if and only if it solves the following MEUE problem:
\vspace{-2pt}
\begin{equation}
\begin{split}
    \min~&\phi(\vp^*), \\
    \text{s.t.}~&\vp^* \in \sP^*.
    \label{eq:meue}
\end{split}
\end{equation}

Problem \eqref{eq:meue} admits a unique solution because its objective function is strictly convex. 
\end{definition}

\subsubsection{Proportionality}
 \citet{bar1999route} found MEUE always satisfies the so-called \emph{proportionality condition}, which dictates ``the same proportions occur for all travelers facing a choice between a pair of alternative segments, regardless of their origins and destinations."  For an illustrative example, consider a 3-node-4-link (3N4L) network shown in Figure \ref{fig:3n4l}, which has four routes connecting the origin (node 1) and the destination (node 3). Route 1 uses links 1 and 3, route 2 uses links 2 and 4, route 3 uses links 1 and 4, and route 4 uses links 2 and 3. 
\begin{figure}[ht]
    \centering
    \includegraphics[width=0.265\textwidth]{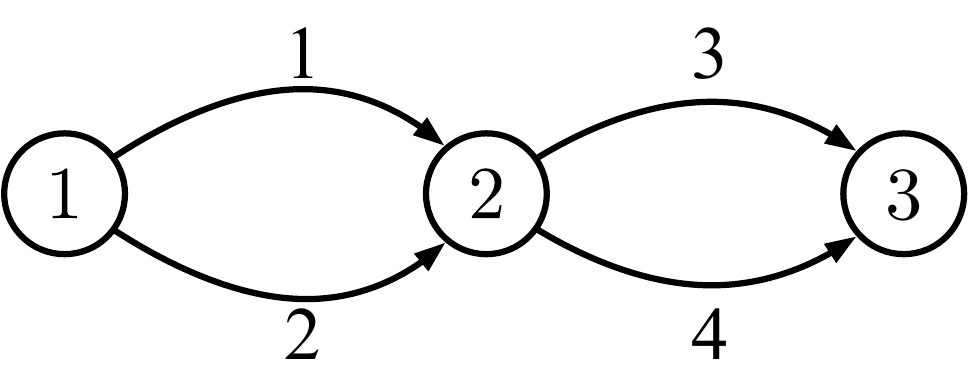}
    \captionof{figure}{A three-node-four-link (3N4L) network.}
    \label{fig:3n4l}
\end{figure}
In this network, a strategy $\vp = [\evp_1, \evp_2, \evp_3, \evp_4] \in \sP$ satisfies the proportionality condition if $\evp_1 / \evp_3 = \evp_4 / \evp_2$, which implies the travelers' choice between the paired alternative segments (link 3 vs. link 4) is irrelevant to their other choices (e.g., link 1 vs. link 2).  \cite{bar2006primal} pointed out the proportionality condition may be used to solve the MEUE problem. This observation has led to the development of highly efficient primal algorithms for the MEUE problem  \citep{bar2010traffic,xie2019new}.  
Despite their success, however, these algorithms are incapable of solving the MEUE problem exactly. This is because satisfying the proportionality condition identified above is not sufficient to find the MEUE route flow \citep{bar2006primal}.  In fact, proportionality between paired alternative segments is but one of many similar conditions the MEUE route flow must obey.  As those higher-order conditions involve complex topology that is much more tedious to identify, \citet{borchers2015traffic} proposed an alternative condition, which we shall call the general proportionality condition in this paper. 
\begin{definition}[General proportionality condition]
    We say a route choice strategy $\vp \in \sP$ satisfies the general proportionality condition if and only if
    \begin{equation}
        \langle \ve, \log(\vp) \rangle = 0, \quad \forall \ve \in \ker{(\mSigma)} \cap \ker{(\mLambda)}.
        \label{eq:general}
    \end{equation}
\end{definition}
To enforce the general proportionality condition, it suffices to identify the basis of $\ker{(\mSigma)} \cap \ker{(\mLambda)}$, which consists of a set of vector $\ve_m$, $m = 1, \ldots, M$ that spans the kernel (i.e.,  $\ker{(\mSigma)} \cap \ker{(\mLambda)} = \spn{(\ve_1, \ldots, \ve_M)}$), and make sure   $\langle \ve_m, \log(\vp) \rangle = 0$ holds for every $m = 1, \ldots, M$. In the literature, $\langle \ve_m, \log(\vp) \rangle = 0$ may be referred to as the $m$-th order proportionality condition. In the 3N4L network, for example, $\ker{(\mSigma)} \cap \ker{(\mLambda)} = \spn{([1, 1, -1, -1]^{\T})}$, i.e., the kernel space can be spanned by a single vector $[1, 1, -1, -1]^{\T}$. Since the kernel space is one dimensional, the general proportionality condition is reduced to the first-order proportionality condition identified by \cite{bar1999route}, i.e., 
\begin{equation}
    \log(\evp_1) + \log(\evp_2) - \log(\evp_3) - \log(\evp_4) = 0, \quad \text{or equivalently,} \quad \evp_1 / \evp_3 = \evp_4 / \evp_2.
\end{equation}

\begin{proposition}[\citet{borchers2015traffic}, Theorem 3.3]
\label{prop:proportionality}
    Under Assumption \ref{ass:monotone},  a UE strategy $\bar \vp^* \in \sP^*$ is the solution to the MEUE problem \eqref{eq:meue} if and only if it satisfies the general proportionality condition \eqref{eq:general}.
\end{proposition}
While this result is a significant step forward, operationalizing the general proportionality condition in an MEUE solution algorithm remains elusive. The challenge is to obtain the basis of the kernel for a sparse matrix in a  computationally viable manner, especially when the matrix contains hundreds of millions of columns. Moreover, it is worth emphasizing that Proposition \ref{prop:proportionality} requires {strict} monotonicity. In fact, it can fail even for a monotone (but not {strictly} monotone) $u(\vx)$. Section \ref{sec:two-conditions} provides such an example.

\subsubsection{MEUE and logit-based stochastic user equilibrium (SUE)}
\label{sec:purification}

SUE may be viewed as the equilibrium of a ``perturbed" routing game in which travelers no longer have access to perfect information.   To describe such information in more general terms, let $\vs\in \sR^{|\sK|}$ be the \emph{valuation} of routes, which depends on the route cost.   In the perturbed game, travelers receive a route valuation littered with a random error $\mathbf{\epsilon}$, which is typically attributed to their imperfect perception.  Subject to this error, the system reaches SUE when every traveler  ``believes" their route choice is the best \citep{daganzo1977stochastic}. Furthermore, when $\mathbf{\epsilon}$ is sampled from a Gumbel distribution, travelers' best response toward route valuation can be described by a logit model \citep{mcfadden1973conditional}. 
Given a scalar $r > 0$, the logit model is a map $q_r: \bar \sR^{|\sK|} \to \sP$ from travelers' route valuation $\vs$ to the corresponding route choice strategy $\vp$, defined as
\begin{equation}
    \evp_k =  \frac{\exp(-r \cdot \evs_k)}{\sum_{k' \in \sK_w} \exp(-r \cdot \evs_{k'})}, \quad \forall k \in \sK.
    \label{eq:logit-w}
\end{equation}
A strategy $\hat \vp \in \sP$ is then defined as a logit-based SUE strategy if it coincides with travelers' choice in response to $c(\vp)$ given by the logit model, i.e., $\hat \vp = q_r(c(\hat \vp))$ \citep{daganzo1977stochastic}.  It is well known \citep[see, e.g.,][]{larsson2001,mamun2011select} that logit-based SUE converges to MEUE when $r\rightarrow \infty$. {To interpret this result, we note that a logit-based SUE with a positive dispersion parameter $r$ has a higher entropy than all UE solutions as long as $r > 0$ (see \citet{mamun2011select} for a proof).  This result is intuitive: the entropy of a route choice pattern is positively related to the number of routes with positive flows. As UE uses a subset of routes while SUE uses all possible routes, it makes sense SUE should have a larger entropy. This relationship, together with the well-known result that SUE converges to UE when $r \to \infty$ \citep{fisk1980some}, indicates that the UE reached by SUE when $r \to \infty$ must be the UE with the highest entropy.}

In theory, this result means one can obtain a solution arbitrarily close to MEUE by solving a logit-based SUE problem with a proper $r$.  In practice, however, few have attempted to solve the MEUE problem this way.  The lack of interest may stem from two main challenges. First, solving the logit-based SUE problem precisely requires enumerating all routes, even those with loops, because, technically,  every route should be used at SUE, however small the probability may be. This is a daunting task on large networks.  Second, it is difficult to determine ex-ante the value of $r$ that guarantees the desired quality of the approximation achieved by this method.  In fact, even measuring the quality of this approximation seems not straightforward --- how do we know an SUE route flow is close enough to the MEUE route flow unless we know how to solve the MEUE problem or at least know how to obtain a tight lower bound?     

Finally, viewing the MEUE route flow as the limit of the SUE flow implies that, to achieve MEUE, travelers must have perfect information since  $r\rightarrow \infty \Rightarrow \epsilon\rightarrow 0$ according to the standard explanation.  Such behavioral perfectionism has been widely criticized in the literature \citep{simon1955behavioral,arrow1966exposition}. Moreover, the interpretation tells us little about how the MEUE route flow might emerge from the evolution of the routing game. 

Therefore, we turn to day-to-day dynamical models for a better behavioral foundation.

\subsection{The CULO model}
\label{sec:culo}

The cumulative logit (CULO) model \citep{li2023wardrop} is a day-to-day (DTD) dynamical model of the routing game.  At its core, CULO consists of two modules: a learning module that updates the route valuation $\vs^t \in \sR^{|\sK|}$ on each day $t$ and a choice module that maps $\vs^t$ to the route choice strategy $\vp^t$. Before the routing game is played, travelers may have a preference for routes, represented by the route valuation $\vs^0$. Those who have no prior information on the routes may simply set $\evs_k^0 =0$ for all $k \in \sK$.
 CULO assumes the travelers incorporate the newly learned route cost $c(\vp^{t-1})$ into the route valuation $\vs^t$ through a weighted cumulative dynamic as follows:
\begin{equation}
    \vs^t = \vs^{t - 1} + \eta^t \cdot c(\vp^{t - 1}),
    \label{eq:culo}
\end{equation}
where the weight $\eta^t$ measures the impact of the cost received on day $t - 1$ on the travelers' valuation on day $t$. Mathematically, the parameter controls how fast the route valuation accrues with the route cost. Behaviorally, it captures how quickly travelers become disposed to ignore the latest information and ``settle down." Thus, $\eta^t$ is referred to as the \emph{proactivity measure}: the larger the $\eta^t$, the more proactive the travelers.

On each day,  a new route choice strategy $\vp^t = q_r(\vs^t)$ is obtained from the latest route valuation,  according to the logit model \eqref{eq:logit-w}. The parameter $r$ in the logit model \eqref{eq:logit-w}, referred to as the \emph{exploration parameter} in CULO,  measures the trade-off between exploration and exploitation: the larger the parameter $r$, the more exploitative the travelers (meaning they are less likely to explore sub-optimal routes). In the CULO model, the parameter $r$ is fixed at a constant value.   One may interpret this setting as travelers' propensity for accepting sub-optimal routes, or their desired balance between exploration and exploitation, is time-invariant.
The following result establishes the global stability of the CULO model --- that is, the convergence to a UE strategy regardless of the initial solution --- under mild requirements for $\eta^t$. Worth noting here is that the weaker of the two conditions only requires $\eta^t$ to be sufficiently small rather than reaching zero at the limit. 
\begin{proposition}[\citet{li2023wardrop}, Theorem 5.4]
\label{thm:convergence-ue}
    Under Assumptions \ref{ass:differentiable}--\ref{ass:monotone}, suppose that $\vs^0 < \infty$, then $\vp^t$ in the CULO model \eqref{eq:culo} converges to a fixed point $\vp^* \in \sP^*$, the solution set to the VI problem \eqref{eq:ue-vi}, if either of the following two conditions is satisfied: (i) $\lim_{t \to \infty} \eta^t = 0$ and $\lim_{t \to \infty} \sum_{i = 0}^t \eta^i = \infty$, or (ii) $\eta^t = \eta < 1/2rL$ for all $t \geq 0$, where $L$ is the Lipschitz constant of $c(\vp)$ (mathematically, any $L \geq \max_{\vp \in \sP} \|\nabla c(\vp)\|_2$ can be used to fulfill the requirement).
\end{proposition}

In this study, we will further explore the relationship between the limiting point of CULO and the initial solution. As we shall see, this relation is the key to unlocking the conditions that ensure the convergence of CULO to the MEUE strategy.

\begin{remark}[Relation with classical DTD models]
A reader familiar with the DTD literature, upon noticing the seemingly striking similarities between CULO and the classical discrete-time DTD models \citep[e.g.,][]{horowitz1984stability}, may question why  CULO converges to UE when other similar models converge to SUE.  This question is addressed at length in \cite{li2023wardrop}. A brief discussion is provided here for the convenience of the reader. 
Let us first consider the DTD model of \citet{horowitz1984stability}, which updates $\vs^t$ as a weighted average of $\vs^{t - 1}$ and $c(\vp^{t - 1})$, i.e., 
\begin{equation}
    \vs^t = (1 - \eta) \cdot \vs^{t - 1} + \eta \cdot c(\vp^{t - 1}).
    \label{eq:watling}
\end{equation}
Variants of the model have been extensively studied in the literature \citep[e.g.,][]{cascetta1993modelling, watling1999stability}, though a fundamental feature remains the same: $\vs^t$ is a weighted \emph{average} of route costs learned over time. Because $\vs^t$ is a weighted average, when $(\vp^t, \vs^t)$ converges to a fixed point $(\hat \vp, \hat \vs)$, we have $ \hat \vs = c(\hat \vp)$ and $\hat \vp = q_r(\hat \vs)$. This leads to $\hat \vp = q_r(c(\hat \vp))$, which implies $\hat \vp$ is a logit-based SUE, with the route valuation at the limit being equal to the route cost.  With a finite exploitation parameter $r$, this model cannot reach UE because, if it does, the travelers would find all UE routes to be \textit{equally good}, and thus choose them with \textit{equal probabilities} (not necessarily a UE strategy). In game theory, this is known as Harsanyi's instability problem \citep{harsanyi1973games}. In the DTD context, the issue was noted in \citet{watling2003dynamics} (Section 3). Once CULO converges to a UE, however, it will be free of this curse. This is because the \emph{cumulative} route costs explain why travelers prefer some routes more than others, as prescribed by the mixed strategy at WE, even though the \emph{present} route costs predict indifference. More specifically, after reaching UE, travelers may have a higher propensity to choose one UE route over another if the former delivers a lower accumulated cost, which may happen when it has a better performance in the past..
We refer the readers to \citet{li2023wardrop}, Section 4.3 for an illustrative example.

\end{remark}

\section{MEUE affirmation conditions}
\label{sec:main}

In this section, we present the main theoretical results concerning the conditions that guarantee the convergence of CULO to the MEUE strategy of the routing game. These conditions will be referred to as the MEUE affirmation conditions.
{Throughout this section, we assume the following conditions always hold.
\begin{itemize}
        \item Assumptions   \ref{ass:differentiable}--\ref{ass:monotone}.
        \item CULO model starts from some initial point $\vs^0 < \infty$ with a fixed and finite exploration parameter $r$ and proactivity parameters $\eta^t$ that satisfy either of the two convergence conditions given in Proposition \ref{thm:convergence-ue}.
    \end{itemize}

}

We begin by presenting a crucial property of the CULO model. 
\begin{lemma}
\label{lm:proportionality}
Starting from any $\vs^0 \in \sR^{|\sK|}$, the CULO model produces a sequence $\{\vp^t\}_{t  = 0}^{\infty}$ that satisfies
$
    \langle \ve, \log(\vp^t) \rangle = -r \cdot  \langle \ve, \vs^0 \rangle
$
for all $\ve \in \ker{(\mSigma)} \cap \ker{(\mLambda)}$.
\end{lemma}
\begin{proof}
    See Appendix \ref{app:proportionality} for detailed proof.
\end{proof}
Lemma \ref{lm:proportionality} implies that for any vector $\ve$ in the basis of $\ker{(\mSigma)} \cap \ker{(\mLambda)},$ the CULO model preserves the value of $\langle \ve, \log(\vp^t)\rangle$ as a constant dependent only on the initial solution. As we shall see, this property is a cornerstone of the results presented in this section. In what follows,  Section \ref{sec:kl-culo} explores the relationship between running CULO and performing KL projection, and Section \ref{sec:two-conditions} gives the conditions under which CULO is guaranteed to reach MEUE.

\subsection{CULO and KL projection}
\label{sec:kl-culo}

Given any two $\vp, \vp' \in \sP$,  the KL divergence between $\vp$ and $\vp'$ can be defined as
\begin{equation}
\label{eq:kl-def}
    D(\vp, \vp') = \langle \diag(\vq) \vp, \log(\vp) - \log(\vp')\rangle. 
\end{equation}

\begin{definition}[The KL projection problem]
    Given any $\vp^0 \in \sP$, the KL projection of $\vp^0$ on $\sP^*$ is defined as
    \begin{equation}
        \bar \vp^* = \argmin_{\vp^* \in \sP^*} D(\vp^*, \vp^0).
        \label{eq:main}
    \end{equation}
\end{definition}
The KL projection problem \eqref{eq:main} is a natural generalization of the MEUE problem \eqref{eq:meue}. Indeed, it reduces to the MEUE problem when $\vp^0 = \vone / \mSigma^{\T} \mSigma \vone$, the equal-distribution route choice that dictates all available routes between each OD pair have an equal probability of being selected. To understand this assertion, it suffices to note that the KL divergence of any $\vp \in \sP$ against the equal-distribution route choice $\vp^0$ reads
\begin{equation}
    D(\vp, \vp^0) = \langle \diag(\vq) \vp, \log(\vp) - \log(\vp^0)\rangle = \phi(\vp) - \langle \diag(\vq) \vp, \log(\vp^0)\rangle = \phi(\vp) + \sum_{w \in \sW} \evd_w \cdot 
    \log(|\sK_w|),
    \label{eq:kl-meue}
\end{equation}
which equals the negative entropy function $\phi(\vp)$ plus a \textit{constant} (the second term). The above relation is well known in the information theory literature \citep{jaynes1957information, kullback1959information}.

The following lemma enables us to check whether a $\bar \vp^* \in \sP^*$ is the solution to the KL projection problem corresponding to an initial solution $\vp^0$.
\begin{lemma}
\label{lm:kl-condition}
    A UE strategy $\bar \vp^* \in \sP^*$ is the KL projection of $\vp^0$ on $\sP^*$ if
    $
        \langle \ve, \log(\bar \vp^*) - \log(\vp^0) \rangle = 0
    $
    for all $\ve \in \ker{(\mSigma)} \cap \ker{(\mLambda)}$.
\end{lemma}
\begin{proof}
    See Appendix \ref{app:kl-condition} for detailed proof.
\end{proof}
We are now ready to present the main result linking the limiting point of CULO to the KL projection of its initial strategy.
\begin{theorem}
\label{thm:main}
  Let $\vp^0$ be an initial strategy and $\vp^*$ be the limiting point of the CULO model corresponding to $\vp^0$. Then  $\vp^*$ is the KL projection of  $\vp^0$ on $\sP^*$.
\end{theorem}
\begin{proof}
    See Appendix \ref{app:main} for detailed proof.
\end{proof}
Theorem \ref{thm:main} may be used to establish several useful properties of the CULO model.

\begin{corollary}
\label{cor:1}
The limiting point of the CULO model is solely determined by the initial strategy $\vp^0$.
\end{corollary}
This property asserts that once the initial point is set, the CULO model will always converge to the same UE strategy if it does converge.  This property ensures the behavioral parameters in CULO --- the exploration parameter $r$ and the proactivity parameter $\eta^t$ --- may not affect the limiting point, even though they clearly have an impact on the evolution path of the dynamical system.   With this property, there exists a stable, one-to-one mapping between the initial and terminal strategies. Otherwise, predicting the terminal strategy would require careful calibration of the behavioral parameters.

\begin{corollary}
\label{cor:2}
The limiting point of the CULO model is continuous with respect to $\vp^0$.
\end{corollary}
This result follows from  Theorem 1.19 in \citet{nagurney2013network}, by recalling that the KL projection problem \eqref{eq:main} is a strictly convex program.
It guarantees a small fluctuation in $\vp^0$ will not result in a large variation in the limiting point. If we only have limited or inaccurate knowledge of $\vp^0$, the property of continuity means that limitation would not be a great concern since it would not cause disproportionately large errors in the predicted outcome of the routing game.

Combining the above two properties with the general convergence condition given in Theorem \ref{thm:convergence-ue} yields the EUC (existence, uniqueness, and continuity of solutions) condition described in \citet{sandholm2005excess}, which is part of what he called the ``desiderata" for an ideal dynamical model needed for equilibrium selection.

To present the third property, let us first denote the set of all routes that may be used by a UE strategy as $\sK^* = \cup_{\vp^* \in \sP^*} \supp(\vp^*)$.  
\begin{corollary}
\label{cor:no-left}
    Suppose $\vp^0 > 0$, i.e., every available route is used by someone at the beginning. Then the limiting point $\bar \vp^*$ of the CULO model satisfies $\supp(\bar \vp^*) = \sK^*$.
\end{corollary}
\begin{proof}
    See Appendix \ref{app:no-left} for detailed proof.
\end{proof}
Corollary \ref{cor:no-left} guarantees the CULO model never excludes a UE route from the set of routes used by the terminal strategy reached at the limit, provided that all routes are initially used.  Thus, the CULO model satisfies the ``no-route-left-behind" policy \citep{bar1999route}, which is a necessary condition for achieving MEUE.

\subsection{Two MEUE affirmation conditions}
\label{sec:two-conditions}

With the results given in the previous section, we are ready to give two conditions that can ensure the limiting point of CULO is MEUE.

\textbf{Condition (A)}. The first condition follows from Theorem \ref{thm:main}, which links the limiting point of CULO to the KL projection, and Equation \eqref{eq:kl-meue}, which asserts that minimizing the KL divergence is equivalent to maximizing entropy against the equal-distribution route choice.

\begin{proposition}
\label{prop:equal}
    If the initial route valuation $\vs^0 = \vzero$ (hence $\vp^0 = \vone / \mSigma^{\T} \mSigma \vone$), then the limiting point of the CULO model is the MEUE strategy. 
\end{proposition}
The initial valuation $\vs^0 = \vzero$ means the travelers have    ``zero information" about the routes initially, hence no preference on any routes can be formed.  This leads to an equal-distribution strategy $\vp^0 = \vone / \mSigma^{\T} \mSigma \vone$. Interestingly, the equal-distribution strategy is the one with the maximum entropy among all $\vp \in \sP$. Hence, when starting from an equal-distribution strategy,  CULO essentially maps $\argmin_{\vp \in \sP} \phi(\vp)$ --- the maximum-entropy strategy --- to  $\argmin_{\vp^* \in \sP^*} \phi(\vp^*)$ --- the MEUE strategy.

\textbf{Condition (B).}  The following result delineates a much larger set of initial strategies that ensure convergence to MEUE. 
\begin{proposition}
\label{prop:main-1}
    If the initial route valuation is formed based on the valuation at the link level, i.e., $\vs^0 = \mLambda^{\T} \vv^0$ for some $\vv^0 \in \sR^{|\sA|}$, then the limiting point $\bar \vp$ of the CULO model is the MEUE strategy. 
\end{proposition}
\begin{proof}
    See Appendix \ref{app:main-1} for detailed proof.
\end{proof}

Thus,  as long as all travelers share the same source of initial link valuations and form their initial route valuation $\vs^0$ (hence the initial strategy) based on that source, the CULO model always converges to the MEUE strategy. %

One is inclined to view Condition (B) as more general than Condition (A) since the former depicts a set containing infinitely many strategies, whereas the latter defines a singleton.  However, it is worth noting Proposition \ref{prop:main-1} relies on Proposition \ref{prop:proportionality}, which in turn requires the link cost function $u(\vx)$ be {strictly} monotone (Assumption \ref{ass:monotone}).  The problem is that strict monotonicity is often violated in real-world applications.  For example, if a link has a flow-independent constant cost, then $u(\vx)$ is monotone but not strictly monotone.  In this case, the condition given in Proposition \ref{prop:main-1} may fail to secure convergence to MEUE for the CULO model, as illustrated in the following counterexample.

\textbf{Counterexample.} Consider a network consisting of three parallel routes, both with constant costs of 1, 1, and 2, respectively. The set of UE strategies is readily described as follows $$\sP^* = \{[\evp_1^*, \evp_2^*, \ \evp_3^* ] \in \sR_+^{3}: \evp_1^* +  \evp_2^* = 1, \evp_3^* = 0\}.$$
Since the network is parallel, it is easy to verify $\ker(\mLambda) \cap \ker(\mSigma)$ is an empty set. As a result, any UE strategy $\vp^* \in \sP^*$ would satisfy the general proportionality condition.  Thus, no matter how we set the initial link valuation $\vv^0 = [\evv_1^0,  \evv_2^0,  \evv_3^0]^{\T}$,  Proposition \ref{prop:main-1} asserts that forming $\vp^0$ based on $\vv^0$ will lead the CULO model to the MEUE strategy --- this must be true because in this case any UE strategy would be considered the MEUE strategy per Proposition \ref{prop:proportionality}.  However, this is reductio ad absurdum since one can easily verify the only  MEUE strategy is $\bar \vp^* = [1/2, 1/2, 0]$. 
The problem here is that both Propositions \ref{prop:proportionality}  and \ref{prop:main-1} fail to hold due to the lack of strict monotonicity. 
Importantly,  Theorem \ref{thm:main} remains valid in this case, and so does Proposition \ref{prop:equal}. We leave it to the reader to verify that if started from $\vs^0 = [0, 0, 0]^{\T}$ (so that $\vp^0 = [1/3, 1/3, 1/3]^{\T}$),  the CULO model will converge to $[1/2, 1/2, 0]^{\T}$, the MEUE strategy.

\section{Exploration of route space}
\label{sec:generation}

Up to this point, we have required that \emph{all} routes be used in the initial strategy to ensure the convergence of CULO --- not only to the MEUE strategy but also to any UE solution (see Proposition \ref{thm:convergence-ue}). However, this requirement is impractical as enumerating all routes is an unbearable computational burden, even for networks of modest size. Nor is it necessary.  In fact, starting from any set that ``covers" the UE route set (covering a set means containing it as a subset) would suffice to secure convergence.  Intuitively, if CULO can reduce an initial strategy using all routes to a strategy only using UE routes, it must be capable of doing the same for an initial strategy using any ``cover" of all UE routes.  In this section, we shall show even predetermining such a cover is unnecessary. Instead, the cover can be ``constructed" iteratively in the evolution of the routing game.   This route generation process may be interpreted as the result of the travelers' exploration of the route space.

We assume travelers start the routing game with a subset of all available routes and on each day $t$ attempt to add to that set the ``best" route discovered on day $t - 1$, provided that route is not already in the set.
In Section \ref{sec:tracing-route}, we prove that CULO equipped with this simple route exploration scheme always converges to a UE strategy.  Yet, the convergence to the MEUE strategy is uncertain due to two complications. First, because the initial strategy no longer encompasses all routes, neither of the two conditions given in Section \ref{sec:two-conditions} seems applicable. Second, the exploration process may not uncover all UE routes.    In Section \ref{sec:tracing-link}, we propose a revised route exploration scheme that promises to resolve these issues.  While the theoretical guarantee can only be partially established, numerical experiments indicate the scheme is an effective heuristic for solving the MEUE problem.

\subsection{Convergence to UE}
\label{sec:tracing-route}

We use $\sK_+^t \subseteq \sK$ to represent the set of routes the travelers actively evaluate on each day and use $\vs_+^t \in \sR^{|\sK_+^t|}$ for the corresponding route valuation. At the end of each day, the travelers between each OD pair $w$ ``discover" the shortest route given the link cost $u(\vx^t)$ observed on that day, say $k^*$. If $k^* \not \in \sK_+^t$, it is added to $\sK_+^{t + 1}$ for possible exploration  on the next day.   Travelers need to initialize the valuation for the new route. This may be done based on past experience, for example, 
\begin{equation}
    s_{k^*}^{t + 1} = \min_{k \in \sK_w \cap \sK_+} \{\evs_k^t\},
    \label{eq:new-val}
\end{equation}
if route $k^*$ is believed to be as good as any route found so far. %
It is worth noting that this initial valuation has little impact on the convergence as long as it is finite. %

Algorithm \ref{alg:route-based} describes the revised CULO model, with the route exploration process described above detailed on Lines \ref{line:generation-start}--\ref{line:generation-end}. 
On Line \ref{line:culoupdate}, we updated the valuation of active routes assuming the proactivity parameter $\eta^t = 1$, which is but one of many possible choices that can ensure convergence.  %

\begin{algorithm}[ht]
\caption{CULO with route exploration and cumulative route valuation.}
\label{alg:route-based}
\footnotesize
\begin{algorithmic}[1]

\State Set $\sK_+^0 \subseteq \sK$ as a subset of routes such that $\sK_+^0 \cap \sK_w \neq \emptyset$ for all $w \in \sW$ and $\vs^0 = \vzero$ (a zero vector with length $|\sK_+|$).
\vspace{2pt}
\For{$t = 0, 1, \ldots$}
\vspace{2pt}
    \State Set $\mLambda_+^t$ and $\mSigma_+^t$ route-link and route-demand incidence matrices  corresponding to $\sK_+^t$. 
    \State Set $\vp_+^t = \vy_+^t / (\mSigma_+^t)^{\T} \mSigma_+^t \vy_+^t$, where $\vy_+^t = \exp(-r \cdot (\mLambda_+^t)^{\T} \vv^t)$. 
    \State Set $\vx^t = \mLambda_+^t \diag(\vq_+^t)  \vp_+^t$ and $\vu^t = u(\vx^t)$, where $\vq_+^t = (\mSigma_+^t)^{\T} \vd$.

    \State Update $\vs_+^{t + 1} = \vs_+^t + \vc_+^t$. \label{line:culoupdate}
    \State Set $\sK_+^{t + 1} = \sK_+^t$.
    \For{all $w \in \sW$}
    \label{line:generation-start}
        \State Find the shortest route $k^*$ based on $\vu^{t}$. 
        \If{$k^*\notin \sK_+^t$}
            \State Add $k^*$ into $\sK_+^{t + 1}$.
            \label{line:add}
            \State Initialize  $s_{k^*}^{t + 1} < \infty$, e.g., following the scheme \eqref{eq:new-val}, and add it to $\vs_+^{t + 1}$ as a new element. 
             \label{line:new-val}
        \EndIf

    \EndFor
    \label{line:generation-end}

\EndFor
\end{algorithmic}
\end{algorithm}

The next result establishes the convergence of Algorithm \ref{alg:route-based} to a UE strategy of the original routing game. 
\begin{proposition}
\label{prop:alg-1}
    By setting the exploration parameter $r$ as a sufficiently small constant in Algorithm \ref{alg:route-based}, 
    the active route set $\sK_+^t$ will converge to a fixed $\overline \sK_+ \subseteq \sK$ and the route choice strategy $\vp_+^t$  will converge to a fixed point $\overline{\vp}_+ \in \overline \sP_+ = \{\vp_+ \in \sR_+^{|\overline \sK_+|}: \overline \mSigma_+ \vp_+ = \vone \}$, where  $\overline \mSigma_+$ is the route-demand incidence matrices corresponding to $\overline \sK_+$. Furthermore, $\overline \vp = [\overline \vp_+; \vzero] \in \sP^*$ (by $\overline \vp = [\overline \vp_+; \vzero]$, we mean a vector in $\sP$ such that $(\evp_k)_{k \in \overline \sK_+} = \overline{\vp}_+$ and $(\evp_k)_{k \in \sK \setminus \overline \sK_+} = \vzero$).
\end{proposition}

\begin{proof}
    See Appendix \ref{app:alg-1} for detailed proof.
\end{proof}
While Algorithm \ref{alg:route-based} always converges to a UE strategy, its convergence to the MEUE strategy is not guaranteed. In part, the problem is caused by the fact that the initial valuation of newly added routes may not always adhere to the general proportionality condition. We address this issue in the next section.

\subsection{Convergence to MEUE}
\label{sec:tracing-link}
As discussed in Section \ref{sec:main}, the convergence to MEUE may be ensured if (i) the routes under travelers' consideration cover all UE routes and (ii) route valuations are obtained from shared link valuations (Condition (B), see Proposition \ref{prop:equal}).  In this section, we discuss how these conditions may be satisfied in the context of route exploration.  
 
Instead of evaluating the newly discovered route in an ad hoc manner, travelers should rely on their past experience of link usage to conform to Condition (B). That is, they anticipate their route experience based on the experience they had on links used by that route.  In order for this initialization scheme to work, the cost accumulation in CULO should occur at the link level.
More specifically, we  assume the travelers keep a record of valuations on links as a vector $\vv^t \in \sR^{|\sA|}$ ($t = 0, 1, \ldots$), and update it using a cumulative scheme similar to \eqref{eq:culo}, i.e., 
\begin{equation}
    \vv^t = \vv^{t - 1} + \eta^t \cdot u(\vx^{t - 1}),
\end{equation}
starting from some $\vv^0 \in \sR^{|\sA|}$.
Based on $\vv^t$, all routes in $\sK_+^t$ can be evaluated --- whether a route is new or old --- as $\vs_+^t = (\mLambda_+^t)^{\T} \vv^t$.%

The new scheme gives rise to Algorithm \ref{alg:link-based}. On Line \ref{line:vv}, we set the proactivity parameter $\eta^t = 1$, similar to Algorithm \ref{alg:route-based}. The route exploration process, described in Lines \ref{ln:start}--\ref{ln:end}, requires no initial valuation of the new route because all route evaluations are performed on Line \ref{line:eval}.

\begin{algorithm}[ht]
\caption{CULO with route exploration and cumulative link valuations.}
\label{alg:link-based}
\footnotesize
\begin{algorithmic}[1]
\State Set $\sK_+^0 \subseteq \sK$ as a subset of routes such that $\sK_+^0 \cap \sK_w \neq \emptyset$ for all $w \in \sW$ and $\vv^0 = \vzero$ (a zero vector with length $|\sA|$).
\vspace{2pt}
\For{$t = 0, 1, \ldots$}
\vspace{2pt}
    \State Set $\mLambda_+^t$ and $\mSigma_+^t$ route-link and route-demand incidence matrices  corresponding to $\sK_+^t$. 
    \State Set $\vp_+^t = \vy_+^t / (\mSigma_+^t)^{\T} \mSigma_+^t \vy_+^t$, where $\vs_+^{t} = (\mLambda_+^t)^{\T} \vv^t$ and $\vy_+^t = \exp(-r \cdot \vs_+^t)$. \label{line:eval} 
    \State Set $\vx^t = \mLambda_+^t \diag(\vq_+^t)  \vp_+^t$ and $\vu^t = u(\vx^t)$, where $\vq_+^t = (\mSigma_+^t)^{\T} \vd$.

    \State Update $\vv^{t + 1} = \vv^t + \vu^t$. \label{line:vv}
    \State Set $\sK_+^{t + 1} = \sK_+^t$.
    \For{all $w \in \sW$}
    \label{ln:start}
        \State Find the shortest route $k^*$ based on $\vu^{t}$. 
        \If{$k^*\notin \sK_+^t$}
            \State Add $k^*$ into $\sK_+^{t + 1}$.
        \EndIf        
    \EndFor
    \label{ln:end}
\EndFor
\end{algorithmic}
\end{algorithm}

If Algorithm \ref{alg:link-based} is initialized from $\sK_+^0 = \sK$ (hence $\mLambda_+^t = \mLambda$ for all $t \geq 0$),  we have
\begin{equation}
    \vs^t = \mLambda^{\T} \vv^t =  \mLambda^{\T} \left( \vv^0 + \sum_{i = 0}^{t - 1} u(\vx^i)\right) = \mLambda^{\T} \vv^0 + \sum_{i = 0}^{t - 1}  \mLambda^{\T} u(\vx^i) = \mLambda^{\T} \vv^0 + \sum_{i = 0}^{t - 1} c(\vp^i),
\end{equation}
which is the accumulated route cost.  Since the validity of  Proposition \ref{prop:alg-1} does not rely on the initial valuation of newly added routes (as long as it is finite),  the convergence to a UE strategy by Algorithm \ref{alg:link-based} can be similarly established. 
We next discuss the conditions under which Algorithm \ref{alg:link-based} converges to the MEUE strategy. 
\begin{proposition}
\label{prop:alg-2}
    Suppose that Algorithm \ref{alg:link-based} converges to a fixed active route set $\overline \sK_+$ and a fixed strategy $\overline{\vp}_+$. If $\overline \sK_+ \supseteq \cup_{\vp^*} \supp{(\vp^*)}$, then  $\overline \vp = [\overline \vp_+; \vzero] \in \sP^*$ must be the MEUE strategy. 
\end{proposition}

\begin{proof}
    See Appendix \ref{app:alg-2} for detailed proof.
\end{proof}
In practice, Algorithm \ref{alg:link-based} cannot always discover a cover of all UE routes, though as we have seen, it can find a cover for the routes used by at least one UE strategy.  A potential remedy is to add some random noises to the current route costs to encourage route exploration.  For example, we may rewrite Line \ref{line:vv} in Algorithm \ref{alg:link-based} as
\begin{equation}
    \vv^{t + 1} = \vv^t + u(\vx^t) + \vepsilon^t, 
    \label{eq:noise}
\end{equation}
where $\vepsilon^t \in \sR^{\sA}$ is a vector of random noises. The variance of $\vepsilon^t$ may vary with $t$, typically starting at a relatively large value (in favor of more aggressive exploration) but gradually decreasing as time proceeds.  Of course, it is difficult to establish any theoretical guarantee for such heuristics, and its performance may vary with problems and parameters.  However, the numerical experiments reported in Section \ref{sec:experiments} will provide preliminary evidence about its effectiveness. 

We close this section by noting that Algorithm \ref{alg:link-based}, in addition to being a behavioral instrument to the proof of convergence, may also be used as a viable alternative to existing specialized algorithms for solving the MEUE problem. 
Implementing Algorithm \ref{alg:link-based} is simple as it requires little more than solving the standard shortest route problem and managing the routes discovered in the dynamical process.  Moreover, it is a strict zeroth-order algorithm, meaning all that is needed to feed into the algorithm is link costs.  Without the need to exploit special problem structures or manipulate complicated graph objects,  Algorithm \ref{alg:link-based} can be quickly implemented to find an approximate solution to the MEUE problem, as well as other non-standard UE routing problems.

\section{Comparison with other dynamical models}
\label{sec:other}

In Sections \ref{sec:kl-culo} and \ref{sec:two-conditions}, we have shown that the CULO model possesses the following properties. 
{
    \begin{itemize}
    \item \textit{Global Stability (GS): the dynamical process converges to a UE strategy regardless of the initial point.    A dynamical process must possess this property to qualify as a behavioral model of UE, i.e., explaining why UE can be reached by reasonable users.  }
    \item \textit{Trajectory Stability (TS): the limiting point of the dynamical process is uniquely determined by its initial point, independent of other parameters integral to the process.  By ensuring the outcome of the dynamical process is not affected by any behavioral contents, TS enhances its robustness.}
    \item \textit{Route Conservation (RC): if the initial point of the dynamical process uses all routes, so does the limiting point.  RC means no route is left behind throughout the dynamical process, a necessary condition of entropy maximization. }
    \item \textit{Proportionality Conservation (PC): if the initial point of the dynamical process satisfies the general proportionality condition, so does the limiting point.  PC is related to RC. The difference is that the general proportionality condition is a sufficient condition for entropy maximization.} 
\end{itemize}
}

\smallskip
{In the literature, there is a group of continuous-time dynamical models of the routing game that are globally stable under Assumptions 
\ref{ass:differentiable}--\ref{ass:monotone}. Given the immensity of the literature on this topic, we shall limit our attention to some of the most well-known models, namely the best-response dynamic \citep{gilboa1991social}, the projection dynamic \citep{friesz1994day, zhang1996local}, the Smith dynamic \citep{smith1984stability}, and the replicator dynamic \citep{taylor1978evolutionary}. A key difference between these models and a discrete-time model like CULO is how the time between two consecutive decision epochs is treated.  
In continuous-time models, this time shrinks to zero, which means travelers' route choice is viewed as ``continuously"  adjustable, and as a result, the potential impact of the rate of this adjustment on convergence is ignored \citep{watling1999stability}. However, whether the model is employed to justify a certain equilibrium as the reasonable outcome of the routing game or develop a solution algorithm for finding such equilibrium, the rate of adjustment cannot be arbitrarily small.  In other words, a continuous model can be ``operationalized" only when it is discretized. 
 Hence, in this section, we discretize these continuous-time models and compare their discrete-time versions with the CULO model in terms of their conformity to the above four properties.} %

{To reveal the mechanism of discretization, let us first present the continuous-time version of the CULO model.}    If both the decision epoch and the proactivity parameter $\eta$ shrink to  0, the CULO model can be written as the following differential equation system
\begin{equation}
\begin{cases}
    \dot \vs = c(\vp), \\
    \vp = q_r(\vs),
    \label{eq:culo-ode}
\end{cases}
\end{equation}
in which $\vs$ increases continuously in time at the \textit{rate} of $c(\vp)$. Accordingly, the original model --- which updates $\vs^{t + 1} = \vs^t + \eta^t \cdot c(\vp^t)$ ---  may be viewed as a numerical solution algorithm for the differential equation \eqref{eq:culo-ode} based on Euler's method \citep[see, e.g.,][for an introduction]{butcher2016numerical}, in which $\eta^t$ may be interpreted as a step size. 
As we shall see, discretizing other continuous-time models may involve parameters playing a similar role as $\eta^t$. %
For simplicity, we shall use the same symbol $\eta^t$ (or $\eta$, if the parameter is a constant) to represent such parameters in the remaining of this section.

\subsection{Best-response dynamic}
\label{sec:bestresponse}

\subsubsection{Description and discretization}
The best-response dynamic \citep{gilboa1991social} assumes travelers  ``receive revision opportunities at a unit rate, and use these opportunities to switch to a current best response" \citep{sandholm2015population}. Given a route choice $\vp \in \sP$, we define $B(\vp) = \argmin_{\vp' \in \sP}\,\langle \vp', c(\vp)\rangle$ as the best response of the travelers given the cost received on the previous day. The best-response dynamic may be written as
\begin{equation}
    \dot{\vp} \in B(\vp) - \vp,
    \label{eq:best-continuout}
\end{equation}
which is a differential inclusion rather than a differential equation, as the best response may not be unique (e.g., multiple minimum cost routes). The best-response dynamic is often used to explain why Nash equilibrium may be reached in finite games (e.g., Rock-Paper-Scissors)  \citep[][Section 13.5.2]{sandholm2015population}. Discretizing  Equation \eqref{eq:best-continuout} using Euler's method gives rise to
\begin{equation}
    \vp^{t + 1}  - \vp^t \in  \eta^t \cdot (B(\vp^t) - \vp^t),
    \label{eq:best-discrete}
\end{equation}
where $\eta^t$ is the step size. To ensure $\vp^{t + 1} \in \sP$, the parameter  $\eta^t$ must be less than 1.

\subsubsection{Properties}

{GS}. Applying the discrete model \eqref{eq:best-discrete} equals solving the routing game with the celebrated Frank-Wolfe algorithm \citep{frank1956algorithm}. It is well known the convergence of that algorithm can be ensured only if the step size decreases progressively at a proper pace  \citep[e.g., setting $\eta^t = 1 / (t + 1)$, as in the so-called method of successive average,][]{nocedal1999numerical}.

TS, RC, and PC.  The model does not satisfy TS even in its continuous-time version. Take the counterexample raised in Section \ref{sec:two-conditions}, where the first two routes have a constant cost of 1, lower than the constant cost of the third route. Hence, if the travelers are initially assigned to route 3, they may end up switching to route 1 or route 2 on the next day, as both give the best response, which means TS is not guaranteed. Moreover, since the limiting point of the model cannot be determined by the initial point, there would be no definitive answers on the adherence to RC and PC either.

\subsection{Projection dynamic}

\subsubsection{Description and discretization}

In the evolutionary game literature, \citet{friesz1994day}'s model and \citet{zhang1996local}'s model are often referred to as the target projection dynamic and the projection dynamic, respectively; see Section 5 in \cite{sandholm2005excess} for an in-depth discussion. Both models were motivated by the projection method for solving routing games \citep{bertsekas1982projection, dafermos1983iterative}. According to this method, the travelers' route choice strategy is updated by 
\begin{equation}
    \vp^{t + 1} = f_{\eta}(\vp^t), \quad \text{where}~f_{\eta}(\vp) = \argmin_{\vp' \in \sP} \|\vp' - (\vp - \eta \cdot c(\vp))\|_2.
    \label{eq:projection}
\end{equation}
The target projection dynamic and the projection dynamic are both derived from Equation \eqref{eq:projection}, though in a different manner. The former fixes $\eta > 0$ and then sets
\begin{equation}
    \dot{\vp} = f_{\eta}(\vp) - \vp = \lim_{\epsilon \to 0} \frac{\vp_{\epsilon} - \vp}{\epsilon}, \quad \text{where}~\vp_{\epsilon} = (1 - \epsilon) \cdot \vp + \epsilon \cdot f_{\eta}(\vp),
\end{equation}
whereas the latter directly lets $\eta \to 0$ in Equation \eqref{eq:projection}, which gives rise to
\begin{equation}
    \dot{\vp} = \lim_{\eta \to 0} \frac{f_{\eta}(\vp) - \vp}{\eta}.
\end{equation}
Therefore, rather than discretizing the two models separately, it may be more natural to directly employ Equation \eqref{eq:projection} as the discrete-time version of these two projection dynamics.

\subsubsection{Properties}

GS. To ensure the convergence of Model \eqref{eq:projection}, the step size $\eta$  may be simply fixed as a sufficiently small constant. As shown by \citet{marcotte1995convergence} (see their Theorem 2.1), a sufficiently small $\eta$ can always ensure the convergence of $\vp^t$ to UE whenever the route cost function $c(\vp)$ is cocoercive, a condition slightly stronger than monotonicity. In particular, when $\nabla c(\vp)$ is symmetric, $c(\vp)$ is cocoercive as long as it is monotone; see Proposition 2.1 in \citet{marcotte1995convergence}.

TS. While we are unable to construct a rigorous proof, we postulate that the discrete model \eqref{eq:projection} is likely to satisfy TS. Specifically, our conjecture is that, given $\vp^0 \in \sP$, the limiting point of the model, denoted as $\bar \vp^* \in \sP^*$, would satisfy
\begin{equation}
    \bar \vp^* = \argmin_{\vp^* \in \sP^*}~\|\vp^* - \vp_0\|_2,
\end{equation}
i.e., the dynamic will reach a UE strategy in $\sP^*$ that minimizes the Euclidean distance from $\vp^0$.  Obviously, the convergence is true if only one iteration is needed before the model converges.  We shall test this hypothesis with numerical experiments but leave a rigorous analysis to a future study. 

RC and PC.  Our reading of the literature does not provide any affirmative answer about these properties. Intuitively, the project dynamic is unlikely to have them because Euclidean projection (as used in the discrete model \eqref{eq:projection}), unlike KL projection, tends to produce sparse solutions \citep{chen2011projection}.

\subsection{Smith dynamic and Replicator dynamic}

\subsubsection{Description and discretization}

We put the Smith dynamic \citep{smith1984stability} and the replicator dynamic \citep{taylor1978evolutionary} together because they are closely related. Below, we first describe the models before turning to the behavioral interpretation. 

{The Smith dynamic} is defined by the following differential equation.
\begin{equation}
    \dot{\evp_k} = \sum_{k' \neq k, k' \in \sK_w} \evp_{k'} \cdot [c_{k'}(\vp) - c_{k}(\vp)]_+ - \evp_k \cdot \sum_{k' \neq k, k' \in \sK_w} [c_{k}(\vp) - c_{k'}(\vp)]_+.
    \label{eq:smith-original}
\end{equation}
By applying Euler's method to Equation \eqref{eq:smith-original}, we obtain a difference equation that reads
\begin{equation}
    \evp_k^{t + 1} - \evp_k^{t} = \sum_{k' \neq k, k' \in \sK_w} p_{k'}^t \cdot \pi_{k', k}^t - \evp_k^{t} \cdot \sum_{k' \neq k, k' \in \sK_w} \pi_{k, k'}^t,
    \label{eq:discrete-time}
\end{equation}
where $\pi_{k, k'}^t = \eta \cdot [c_{k}(\vp^t) - c_{k'}(\vp^t)]_+$ ($\eta > 0$ is the step size).

\textbf{The replicator dynamic} has many equivalent forms \citep[see, e.g.,][Example 13.6]{sandholm2015population}, one of which reads
\begin{equation}
    \dot{\evp_k} = \sum_{k' \neq k, k' \in \sK_w} \evp_{k'} \cdot \evp_k \cdot [c_{k'}(\vp) - c_{k}(\vp)]_+ - \evp_k \cdot \sum_{k' \neq k, k' \in \sK_w} \evp_{k'} \cdot [c_{k}(\vp) - c_{k'}(\vp)]_+.
    \label{eq:replicator-new}
\end{equation}
First suggested by \citet{schlag1998imitate},  Equation \eqref{eq:replicator-new} is also known as the proportional pairwise comparison dynamics. %
By applying Euler's method to Equation \eqref{eq:replicator-new}, we readily obtain a difference equation
\begin{equation}
    \evp_k^{t + 1} - \evp_k^{t} = \sum_{k' \neq k, k' \in \sK_w} \evp_{k'}^t \cdot \gamma_{k', k}^t - \evp_k^{t} \cdot \sum_{k' \neq k, k' \in \sK_w}  \gamma_{k, k'}^t,
    \label{eq:discrete-replicator}
\end{equation}
where $\gamma_{k, k'}^t = \eta \cdot \evp_{k'}^t \cdot [c_{k}(\vp^t) - c_{k'}(\vp^t)]_+$ ($\eta$ is  the step size). %

\textbf{Behavior interpretation.} 
On each day $t$, if the probability of a traveler switching from their current route $k \in \sK_w$ to a \textit{different} route $k' \in \sK_w$ is set as $\pi_{k, k'}^t$, then the first and the second terms in Equation \eqref{eq:discrete-time} represent, respectively, the proportion of travelers switching from other routes to route $k$ and that from route $k$ to other routes. The same interpretation applies to Equation \eqref{eq:discrete-replicator} by replacing $\pi_{k, k'}^t$ with $\gamma_{k, k'}^t$.
In both interpretations, the probability of the traveler sticking to  their original choice $k$ is one less the total probabilities of changing to other routes, i.e., $\pi_{k,k}^t := 1 - \eta \cdot \sum_{k' \neq k, k' \in \sK_w} [c_{k}(\vp^t) - c_{k'}(\vp^t)]_+ \quad$ for the Smith Dynamic and $\gamma_{k,k}^t := 1 - \eta \cdot \sum_{k' \neq k, k' \in \sK_w} \evp_{k'}^t \cdot [c_{k}(\vp^t) - c_{k'}(\vp^t)]_+$ for the replicator dynamic. 

To ensure these probabilities are non-negative, $\eta$ must be sufficiently small. Here, we note that continuous-time models implicitly assume $\eta = 0$, and hence, the feasibility constraint can always be secured. Behaviorally, the smaller the value of $\eta$, the less willing the traveler is to explore new routes.

\textbf{Comparison.} The two models are almost identical, except for the factor $\evp_{k'}^t$ added before $[c_{k}(\vp^t) - c_{k'}(\vp^t)]_+$ by the replicator dynamic to scale the switching probability.  \citet{schlag1998imitate} explains the scalar as follows. Suppose travelers can only observe the cost of the route they take but are allowed to gather route information from a randomly picked fellow traveler.  Then the scalar $\evp_{k'}^t$ may be interpreted as the probability of the random traveler taking route $k'$. To understand how the scalar makes a difference,  consider the probability that a traveler currently on route $k$ switches to a new route $k'$ on day $t$, which nobody selected on that day (hence $\evp_{k'}^t = 0$).  Under the Smith dynamic, the switching probability would be $\eta \cdot [c_{k}(\vp^t) - c_{k'}(\vp^t)]_+$, which is positive as long as the cost of route $k'$ is strictly lower than that of route $k$. In contrast, the switching probability given by the replicator dynamic is $\eta \cdot \evp_{k'}^t \cdot [c_{k}(\vp^t) - c_{k'}(\vp^t)]_+ = 0$. The rationale behind the replicator dynamic is that, as the traveler has nowhere to learn about the better route $k'$, they would have no chance to take it. On the other hand, the Smith dynamic would better fit the situation where every traveler has access to full information all the time.

\subsubsection{Properties}

GS.  It is straightforward to show that the discrete version of either model converges to UE when $\eta$ is fixed as a sufficiently small constant.

TS, RC, and PC.  We shall show the replicator dynamic and CULO are equivalent in continuous time, which might shed light on the properties of the former. Indeed, differentiating the second line $\vp = q_r(\vs)$ in Equation \eqref{eq:culo-ode} with respect to time yields 
\begin{align}
    \frac{\dot{\evp_k}}{r} &=
    - 
    \frac{\exp(-r \cdot \evs_{k})}{\sum_{k' \in \sK_w} \exp(-r \cdot \evs_{k'})} \cdot \left(\dot \evs_k - \sum_{k' \in \sK_w} \frac{\exp(-r \cdot \evs_{k'})}{\sum_{k' \in \sK_w} \exp(-r \cdot \evs_{k'})} \cdot \dot \evs_{k'}  \right) \nonumber \\ &= - \evp_k \cdot \left(\dot \evs_k - \sum_{k' \in \sK_w} \evp_{k'} \cdot \dot \evs_{k'}\right)
    \nonumber = -\evp_k \cdot \left(c_k(\vp) - \sum_{k' \in \sK_w} \evp_{k'} \cdot c_{k'}(\vp)\right) = -\evp_k \cdot \sum_{k' \in \sK_w} \evp_{k'} \cdot( c_k(\vp) -  c_{k'}(\vp))\nonumber \\[5pt]
    &=\evp_k \cdot \sum_{k' \neq k, k' \in \sK_w} \evp_{k'} \cdot  [c_{k'}(\vp) - c_{k}(\vp)]_+ - \evp_k \cdot \sum_{k' \neq k, k' \in \sK_w} \evp_{k'} \cdot [c_{k}(\vp) - c_{k'}(\vp)]_+.
    \label{eq:revelation}
\end{align}
The reader can verify that Equations \eqref{eq:revelation} and \eqref{eq:replicator-new} are identical except for a re-scaling of time by $r$. This revelation is surprising as the two DTD models have distinct behavior mechanisms in their respective discrete forms --- one based on the logit model while the other based on pairwise route switching --- and have not been previously connected with each other. Yet, the above analysis indicates they are closely related when the decision epoch shrinks to zero.

Based on the above finding, we postulate that the behavior of the discrete replicator dynamic \eqref{eq:discrete-replicator} may be similar to that of CULO if a sufficiently small step size $\eta$ is adopted.  Numerical experiments presented in the next section will show the model tends to (i) satisfy RC if $\eta$ is sufficiently small and (ii) satisfy PC approximately when $\eta \to 0$, but uncovers no evidence confirming its compliance with TC.  A thorough theoretical investigation of this model and other discrete models discussed in this section is left to a future study.

\section{Numerical results}
\label{sec:experiments}
To validate the analysis results presented in the previous sections, numerical experiments are performed on two networks: the 3N4L, as shown earlier in Figure \ref{fig:3n4l}, and the Sioux-Falls network \citep{leblanc1975algorithm}, which has 24 nodes, 76 links, and 528 OD pairs.  For a route choice strategy $\vp$, we use the relative gap of its corresponding link flow $\vx \in \sX = \{\vx: \vx = \bar \mLambda \vp, \ \vp \in \sP\}$, denoted as $\delta(\vx)$, to measure its distance from WE. The relative gap is computed by
\begin{equation}
    \delta(\vx) =-\frac{\langle u(\vx), \vx' - \vx \rangle}{\langle u(\vx), \vx \rangle}, \quad \vx' \in \argmin_{\vx'' \in \sX}~\langle u(\vx), \vx''  \rangle.
\end{equation}
A solution is accepted as a UE solution whenever $\delta$ is smaller than a predefined value, taking a default of $10^{-5}$ in this study.      Unless otherwise stated, we also fix the proactivity parameter $\eta^t$ in the CULO model at 1 in all experiments.  We next provide some details of the two networks.

\textbf{3N4L.} The number of travelers from node 1 to node 4 is 10. Given the flow $\evx_a$ on link $a$, we model its costs as $u_a = \evh_a + \evw_a \cdot \evx_a^4$, where $[h_1, h_2, h_3, h_4]^{\T} = [4, 20, 1, 30]^{\T}$ and $[\evw_1, \evw_2, \evw_3, \evw_4]^{\T} = [1, 5, 30, 1]^{\T}$. Under this setting, the set of UE strategies can be written as
\begin{equation}
    \sP^* = \{\vp^*: \vp^* = [0.3 - \lambda, 0.4 - \lambda, 0.3 + \lambda, \lambda]^{\T}, \ \lambda \in [0, 0.3] \}.
    \label{eq:ue-set}
\end{equation}
It can be verified that $\bar \vp^* = [0.18, 0.28, 0.42, 0.12]^{\T}$ is the MEUE strategy, which corresponds to $\lambda = 0.12$. In our experiments, once a UE strategy $\vp^* \in \sP^*$ is found, the corresponding $\lambda(\vp^*)$  is computed as follows:
\begin{equation}
    \lambda(\vp^*) = [(0.3 - \evp_1^*) + (0.4 - \evp_2^*) + (\evp_3^* - 0.3) + \evp_4^*] / 4. 
    \label{eq:lambda}
\end{equation}

\textbf{Sioux-Falls.} We refer the readers to \citet{leblanc1975algorithm} for the topology, travel demand, and cost function of the Sioux-Falls network. A highly sophisticated MEUE algorithm developed by \citet{feng2022bush} --- which promises to obtain a solution with close-to-float precision --- is employed to produce the benchmarks.  The MEUE route flow for the Sioux-Falls network found by their algorithm contains  770 routes, with an entropy of 59235.10.

\subsection{Convergence of CULO toward MEUE}
\label{exp:culo}

In Section \ref{sec:exp-1}, we run CULO with randomly generated initial points and examine the distribution of the limiting points.   We then compare the entropy values of initial and limiting points (Section \ref{sec:initial-limiting}). Finally, Section \ref{sec:practical} tests a CULO-based algorithm equipped with route discovery.

\subsubsection{Distribution of CULO's limiting points}
\label{sec:exp-1}
In this experiment, a set of initial points are randomly selected for the 3N4L network to run the CULO model.  Two strategies are employed to generate the initial points. In the first, we sample $\vp^0$ from a uniform distribution and re-scale $\vp^0$ to fit the flow conservation condition. We then choose $\vs^0 = - \log(\vp^0) / r$ such that $\vp^0$ would be reproduced from the route choice function $q_r(\vs^0)$. This strategy guarantees all $\vp^0 \in \sP$ have an equal chance to be selected. Rather than sampling $\vp^0$ directly, the second strategy samples $\vs^0$ from a normal distribution centered at $\vzero$ -- thus, the initial points around $\vs^0 = \vzero$ would have a greater chance to be selected. %
In both cases, the sample size is set to 5000, and the equal-distribution initial point, $\vs^0 = [0, 0, 0, 0]^{\T}, \vp^0 = [1/4, 1/4, 1/4, 1/4]^{\T}$, is employed as a benchmark. For each initial point, we run CULO until convergence and then invoke Equation \eqref{eq:lambda} to obtain the corresponding $\lambda$. 
\begin{figure}[ht]
    \vskip 0.1in
    \centering
    \begin{subfigure}[b]{0.45\textwidth}
        \centering
        \includegraphics[height=0.45\columnwidth]{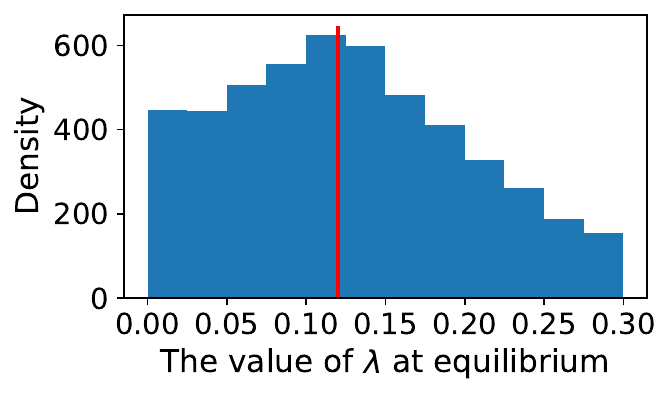}
        \caption{$\vp^0$ sampled from a uniform distribution.}
        \label{fig:i-1}
    \end{subfigure}
    \begin{subfigure}[b]{0.45\textwidth}
        \centering\includegraphics[height=0.45\columnwidth]{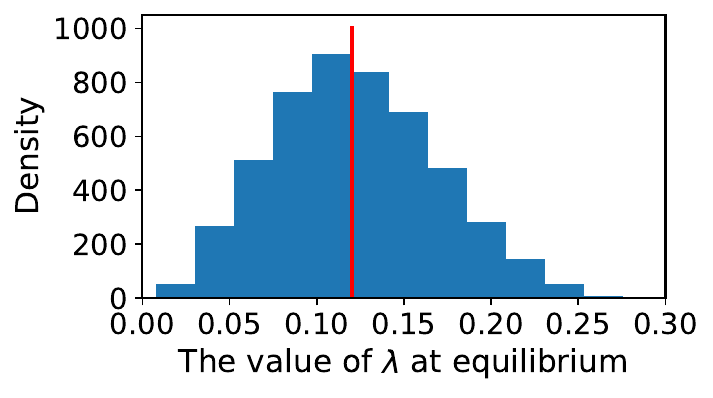}
        \caption{$\vs^0$ sampled from a normal distribution.}
        \label{fig:i-2}
    \end{subfigure}
    \caption{Distribution of $\lambda$ corresponding to UE strategies of the 3N4L network, obtained from 5000 different initial points by CULO. The red line highlights the $\lambda$ value corresponding to the equal-distribution initial point. }
    \label{fig:exp-i}
\end{figure}

  Figure \ref{fig:exp-i} plots, for each initialization strategy, the histogram of $\lambda$ values corresponding to the five thousand UE strategies. As expected, when $\vp^0$ is sampled from a uniform distribution, $\lambda$ spreads over the entire theoretical range ($[0, 0.3]$), whereas it concentrates around the MEUE strategy ($\lambda = 0.12$) when a normal distribution is used to sample $\vs^0$.

Per Proposition \ref{prop:equal},  CULO is guaranteed to reach the MEUE strategy if started from the equal-distribution initial point. Our results confirm that this is indeed the case: the vertical red line in the plots is the solution found by CULO when $\vs^0 = [0, 0, 0, 0]^{\T}$.  A more interesting finding, however, is that the MEUE strategy aligns perfectly with the peak of the histogram in both cases despite the vastly different sampling methods. {
    The result provides an interesting confirmation that the MEUE strategy is indeed the most likely outcome of the routing game, no matter how we choose to initialize it.}

\subsubsection{Relation between initial and limiting entropy}
\label{sec:initial-limiting}

We  proceed to compare $-\phi(\vp^0)$, the entropy at the initial point, with $-\phi(\vp^*)$, the entropy at $\bar \vp^* = \lim_{t \to \infty} \vp^t$. Recall that CULO always guides the initial strategy with the highest entropy (equal-distribution strategy) to the MEUE strategy, which implies the entropy of $\vp^0$ and that of $\bar \vp^*$ may be positively correlated. However, since UE is a more ``orderly" state compared to a non-equilibrium state, we expect the entropy of $\bar \vp^*$ to be lower than that of $\vp^0$. 

To validate our hypotheses, we run experiments in the 3N4L network by initializing $\vs^0$ with two strategies. The first directly generates $\vs^0$ from a normal distribution, rather like the second strategy in Section \ref{sec:exp-1}. The second strategy  first randomly generates $\vv^0$ --- travelers' initial valuation of all available links --- and sets $\vs^0 = \mLambda^{\T} \vv^0$. This way, $\vp^0$ always satisfies the general proportionality condition. For each initialization strategy, the sample size is set as 250. 
\begin{figure}[ht]
    \vskip 0.1in
    \centering
    \begin{subfigure}[b]{0.42\textwidth}
        \centering
        \includegraphics[height=0.5\columnwidth]{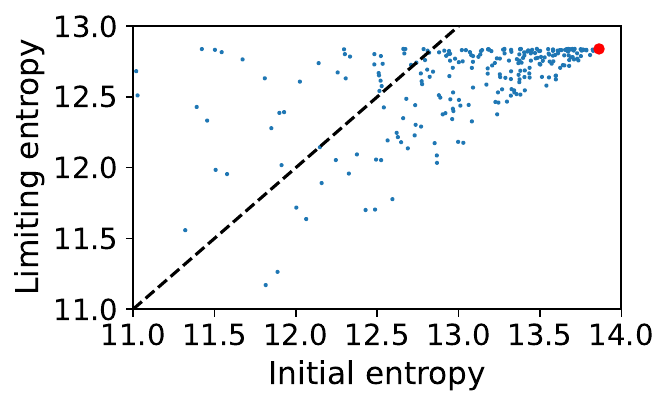}
        \caption{$\vs^0$ randomly generated from a normal distribution.}
        \label{fig:ii-1}
    \end{subfigure}
    \hspace{5pt}
    \begin{subfigure}[b]{0.42\textwidth}
        \centering
        \includegraphics[height=0.5\columnwidth]{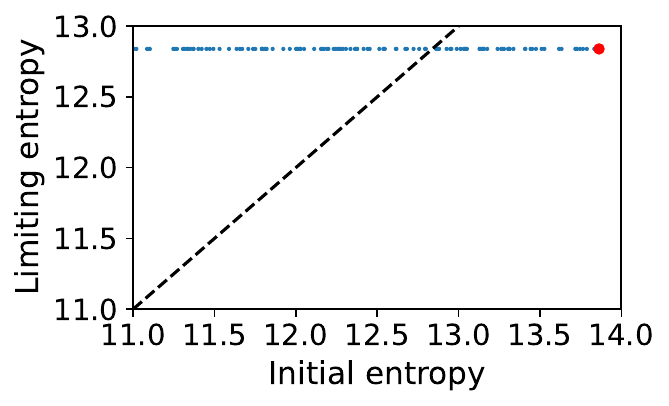}
        \caption{$\vv^0$ randomly generated from a normal distribution.}
        \label{fig:ii-2}
    \end{subfigure}
    \caption{Initial entropy v.s. limiting entropy for 250 samples of the 3N4L network. The red point highlights the pair corresponding to the equal-distribution initial point (i.e., $\vs^0 = \vzero$); the black dashed line is the 45-degree line.}
    \label{fig:exp-ii}
\end{figure}

The scatter plots of all samples --- the coordinates of a point are $(-\phi(\vp^0), - \phi(\vp^*))$ for a given sample --- are reported in Figure \ref{fig:exp-ii}.
First and foremost, the red point is always located at the top right corner in both plots, which validates Proposition \ref{prop:equal}: starting from the maximum-entropy strategy, CULO  converges to the MEUE strategy.
When $\vs^0$ is directly generated from a normal distribution (Figure \ref{fig:ii-1}), there is a clear positive correlation between the limiting entropy and the initial entropy.  Also, most points (about 83.2\%) lie beneath the 45-degree line, indicating that entropy tends to decrease in the equilibrium-finding process. Both observations are well aligned with the expectation from our analysis. 
When $\vs^0$ is obtained from randomly generated $\vv^0$, the limiting entropy of all initial points should reach the maximum possible value, as established in Proposition \ref{prop:main-1}. Figure \ref{fig:ii-2} confirms this theoretical prediction.  Interestingly, the vast majority of the data pairs,  80.4\%,  are now located above the 45-degree line. Thus, in this case, the entropy tends to increase in the equilibrium-finding process. A possible explanation is that the second initialization strategy drew initial solutions disproportionately from the regions associated with lower entropy values. 
We leave an in-depth look into this phenomenon to future studies.

\subsubsection{Route discovery strategies}
\label{sec:practical}
We run Algorithms \ref{alg:route-based} and \ref{alg:link-based} on the Sioux-Falls network to test the performance of different route discovery strategies. Four scenarios, labeled Scenarios (A)-(D), are examined. Scenario (A) is the benchmark, which employs a predetermined route set containing 1238 routes, including all 770 UE routes found using the aforementioned algorithm \citep{feng2022bush}. In this scenario, no route exploration is needed, and the standard CULO algorithm is executed. In the other three scenarios, the route set is initially populated with the shortest route for each O-D pair (with the link cost set to zero).  Scenario (B) tests Algorithm \ref{alg:route-based}, in which the valuation of a new route is initialized using Equation \eqref{eq:new-val}.  Scenarios(C) and (D) both test Algorithm \ref{alg:link-based}. The difference is that Scenario (D) enhances the exploration by adding random noise to link costs (as described in Equation \eqref{eq:noise}). In the implementation, we also gradually reduce the variance of the error term $\epsilon_t$ at a rate of $\mathcal{O}(1/t)$. We stop adding noises into link costs when no new routes are found in a sufficiently long time,

\begin{figure}[ht]
\vspace{6pt}
\centering
\includegraphics[width=0.8\textwidth]{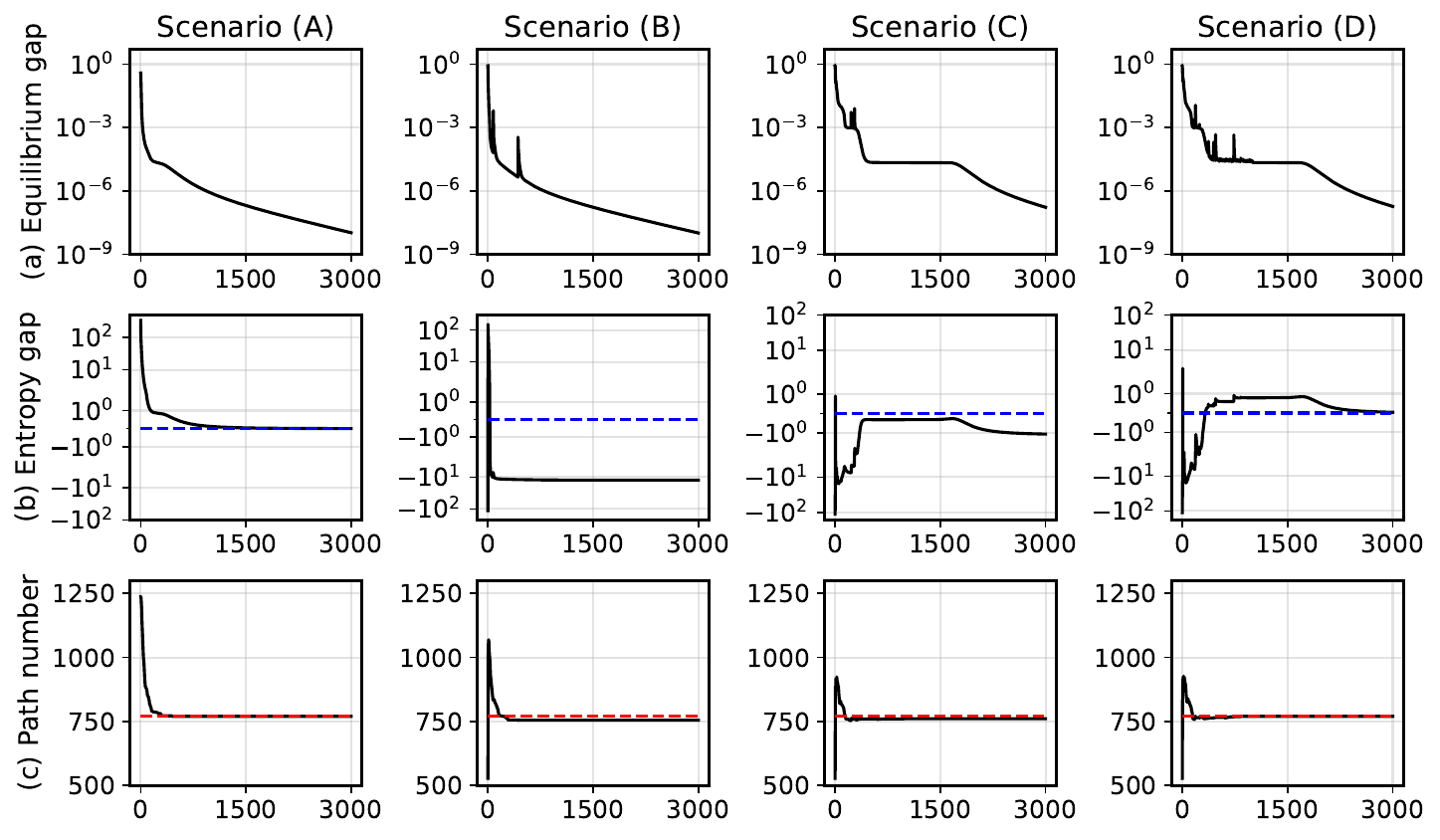}
\caption{Convergence patterns of CULO for the Sioux-Falls network in four scenarios. Scenario (A): CULO with predetermined routes. Scenario (B): Algorithm \ref{alg:route-based}. Scenario (C): Algorithm \ref{alg:link-based}. Scenario (D): Algorithm \ref{alg:link-based} with exploration noises.
In each column, plot (a) reports the relative equilibrium gap, (b) reports the difference in entropy between the CULO solution and the baseline solution,  normalized by the number of OD pairs and plotted in symlog scale (where the blue dashed line corresponds to a gap of zero); and  (c) reports the number of routes actively used by travelers (where the red dashed line corresponds to the number of routes contained in the benchmark solution). }
\label{fig:experiment}

\vspace{6pt}
\end{figure}
Figure \ref{fig:experiment} compares the convergence patterns of the CULO dynamical process in the four scenarios. As anticipated by our analysis results, CULO converges smoothly to the MEUE strategy in Scenario (A) in terms of both the entropy value and the UE route set. Compared to specialized traffic assignment algorithms such as TAPAS \citep{bar2010traffic} and bush-based algorithms \citep{nie2010class}, its convergence is relatively slow:  the relative gap remains above $10^{-9}$ after 3000 days (more than eight years).  However, to reach a relative gap of about $10^{-5}$, CULO only requires about 1--2 months. 

Neither Scenario  (B) nor (C) is able to converge to the MEUE strategy.  In both cases, the route exploration process ended up missing a small number of UE routes and, as a result, produced solutions with entropy values markedly lower than the benchmark.  It is worth noting that they had no problem converging to a UE strategy, although their convergence path is not as smooth as in Scenario (A).   With the help of exploration noises, Scenario (D) successfully discovered all routes contained in the benchmark solution and obtained a high-quality approximation to the MEUE strategy. However, the ``randomized" route discovery process slowed down convergence, a price one has to pay in order to increase the likelihood of identifying all UE routes. Also, while the strategy succeeded in finding all UE routes for this problem, there is no guarantee it will for other problems.

\subsection{Comparison with other dynamical models}
\label{sec:comparision}

In this section, we numerically investigate the properties of the four DTD models discussed in Section \ref{sec:other} (best-response, projection, replicator, and Smith) and compare them with CULO.  We begin with the 3N4L network (Sections \ref{sec:exp-c-1}) and turn to the Sioux-Fall network in Section \ref{sec:exp-c-3}.

\subsubsection{3N4L network}
\label{sec:exp-c-1}

Our focus is on the effect of the step size on the limiting point of each model.  Based on trial and error, we set the range of the step size $\eta$ in our experiments as follows:
\begin{itemize}
    \item CULO: Set $r = 1$, fix $\eta^t$ as a constant $\eta$ in Equation \eqref{eq:c_plus}, and test $\eta = 0.05, 0.10, \ldots, 1$.
    \item Best-response: Set $\eta^t = \eta / (1 + t)$ in Equation \eqref{eq:best-discrete} and test $\eta = 0.05, 0.10, \ldots, 0.95$.
    \item Projection: Set $\eta = 0.02, 0.04, \ldots, 0.2$ in Equation \eqref{eq:projection}.
    \item Smith: Set $\eta = 0.005, 0.0010, \ldots, 0.13$ in Equation \eqref{eq:discrete-time}.
    \item Replicator:  Set $\eta = 0.02, 0.04, \ldots, 0.4$ in Equation \eqref{eq:discrete-replicator}.
\end{itemize}
Thus, for all models listed above, their performance is dictated by $\eta$.  In all runs, the initial point is fixed as $\vp^0 = [0.25, 0.25, 0.25, 0.25]^{\T}$. We terminate  CULO, Smith, and replicator when the equilibrium gap reaches $\delta = 10^{-10}$. For best-response and projection, the convergence criterion is relaxed to $\delta = 10^{-5}$ because aiming for a higher precision would be too time-consuming for these two dynamics.  Figure \ref{fig:exp-iii} reports the results, including (a) the value of $\lambda$ corresponding to the UE strategy reached by the model, calculated based on Equation \eqref{eq:lambda} (the top plot) and (b) the number of iterations required to achieve a satisfactory convergence (the bottom plot).

\begin{figure}[ht]
\vspace{6pt}
\centering
\includegraphics[width=0.9\textwidth]{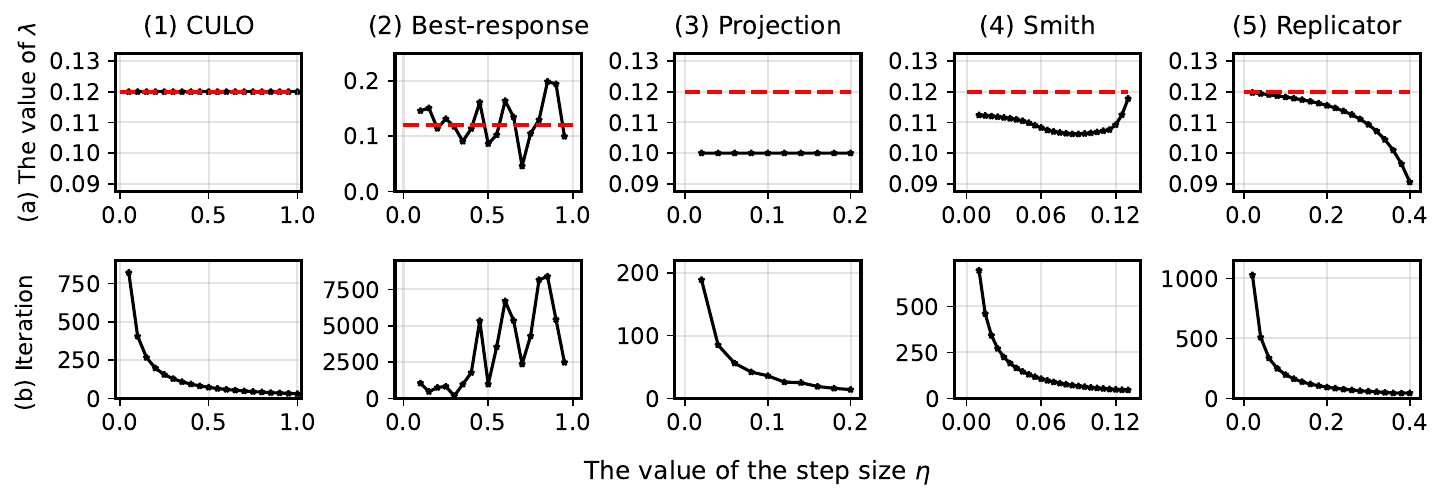}
\caption{The relationship between the limiting points of different models with respect to their step sizes. In each column, plot (a) reports the value of $\lambda$ corresponding to the UE strategy reached by the model (the red dashed line highlights the corresponding value of the MEUE); plot (b) reports the number of iterations required for convergence.}
\label{fig:exp-iii}

\vspace{6pt}
\end{figure}

First, with the exception of the best-response dynamic, a larger $\eta$ always accelerates convergence in the tested range. For the best-response dynamic, the opposite is true: as $\eta$ increases, the number of iterations required for convergence generally trends up, though the relationship is not monotonic. 
CULO, as guaranteed by Corollary \ref{cor:1}, always reaches the MEUE strategy (with $\lambda$ equal to 0.12) regardless of the value of $\eta$. The projection dynamic is the only other model whose limiting point is not affected by $\eta$, hinting compliance with TS.  Upon close examination, we also confirmed that its limiting point is indeed the Euclidean projection of the initial point onto the equilibrium set.  The other three models fail to meet TS, as their limiting point all changes with $\eta$.  The limiting point of the best-response dynamic oscillates abruptly around the MEUE strategy. For the Smith dynamic and the replicator dynamics, their limiting point seems to always stay on one side of the MEUE strategy (i.e., $\lambda \le 0.12$) and varies much more smoothly with $\eta$. The result also appears to confirm our conjecture that the replicator dynamic tends to converge to the MEUE strategy when $\eta \to 0$.

Could the replicator dynamic be used as an  MEUE problem solver? The answer is probably yes if one is willing to tolerate the slow convergence associated with the use of a very small step size. When $\eta = 0.02$, the replicator dynamic finds a high-quality MEUE approximation after more than 1000 iterations. For $\eta = 0.4$, the convergence takes only 44 iterations, but the limiting point drifts far away from the MEUE strategy.  CULO does not face this dilemma, thanks to the theoretical guarantee.  When $\eta = 1$, it converges in 30 iterations, and the limiting point is still the MEUE strategy.

To recapitulate, our numerical results show (i) all models satisfy GS with a properly selected step size; (ii) no model other than CULO and the projection dynamic may satisfy TS; and (iii) no model other than CULO and the replicator dynamic may satisfy PC. Here, we note that other models sometimes produce solutions close to MEUE, but we tend to believe these occurrences as coincidental rather than a consistent pattern. We next turn to these models' adherence to RC, for which we need to use the Sioux-Falls network.

\subsubsection{Sioux-Falls network}
\label{sec:exp-c-3}

In the experiment, we run the models from an equal-distribution initial strategy using all 770 UE routes and check their convergence patterns, particularly whether any of the routes will be eliminated when a UE strategy is reached.  A route is considered ``eliminated" (i.e., not used by anyone) if the proportion of the travelers selecting it is less than $\tau$. We test two values of $\tau$: $10^{-4}$ and $10^{-6}$. The step size for each model is appropriately tuned such that the relative gap gradually converges to zero as fast as possible. We set the convergence criterion  $\delta = 10^{-6}$ in this experiment. The results are reported in Figure \ref{fig:exp-sf}, including the detailed convergence pattern for (a) the relative gap $\delta(\vp^t)$; (b) the entropy $\phi(\vp^t)$; (c) the number of used routes, i.e., the size of the set $\{k: \evp_k^t > \tau\}$; (d) the violation of the first- and second-order proportionality condition, measured by $\langle \ve_{i}, \log(\vp^t) \rangle$ ($i = 1, 2$), where $\ve_1$ and $\ve_2$ are the first and second basis of $\ker{(\mLambda)} \cup \ker{(\mSigma)}$).

\begin{figure}[ht]
\vspace{6pt}
\centering
\includegraphics[width=0.95\textwidth]{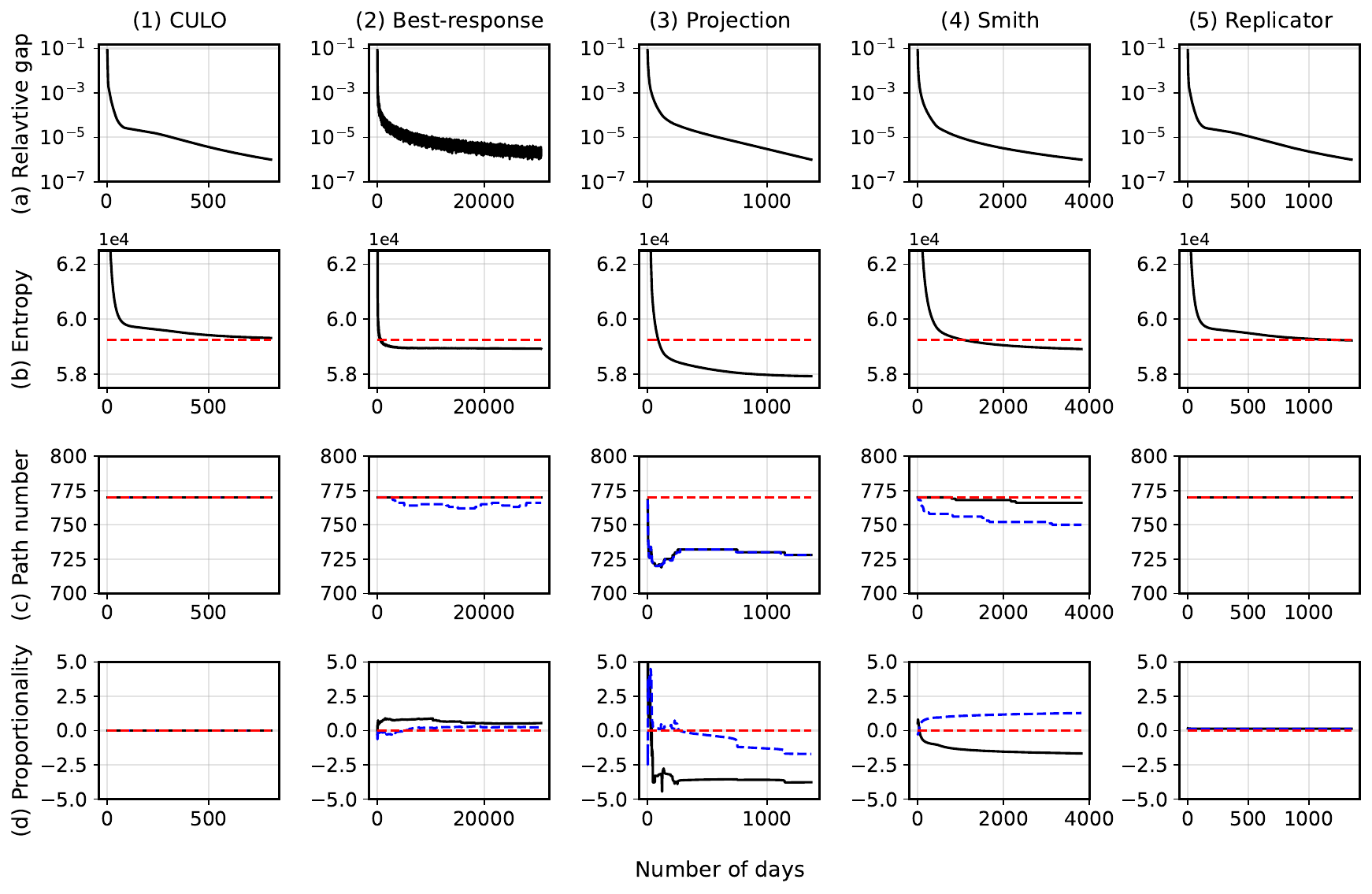}
\caption{Convergence patterns of the five models. In each column, plot (a) reports the relative gap; plot (b) reports entropy; plot (c) reports the number of used routes (the blue dashed line and the black solid line correspond to the number of routes used by more than $10^{-2}$ and $10^{-4}$ percent of the travelers, respectively); plot (d) reports the violation of the first- and second-order proportionality condition (the black solid line and the blue dashed line correspond to the values of $\langle \ve_{1}, \log(\vp^t) \rangle$ and $\langle \ve_{2}, \log(\vp^t) \rangle$ respectively.)}
\label{fig:exp-sf}

\vspace{6pt}
\end{figure}

Plot (a) in Figure \ref{fig:exp-sf} concerns global stability (GS). It confirms all models satisfy GS, i.e., they converge to a satisfactory UE solution. To reach the convergence threshold, CULO requires the least number of iterations (about 800), followed by the projection dynamic and the replicator dynamic, both taking roughly twice as many iterations to converge as needed by CULO.  The slowest is the best-response dynamic, which needs 30,000 iterations to reach $10^{-6}$, at least an order of magnitude slower than any other models.  This is hardly surprising if one recalls that the Frank-Wolfe algorithm --- notorious for its painfully slow convergence due to zigzagging behavior --- is, in fact, a variant of the best-response dynamic.  The Smith dynamic fares far better than the best-response dynamic but falls behind the other three.

Plot (c) examines Route Conservation (RC).  As seen from plot (c), both CULO and the replicator dynamic left no UE routes behind (all 770 routes are still used when equilibrium is reached), for both tolerance values ($\delta$).  The best-response dynamic kept all 770 routes when $\delta$ is  $10^{-6}$, but left a few out if $\delta = 10^{-4}$.  The Smith dynamic kept only 750 routes even with the looser tolerance standard ($\delta =10^{-4}$), but the projection dynamic is the worst in this regard: it eliminated almost 40 routes from the UE set. To be sure,  it is possible that a route considered eliminated even by the more stringent standard may still be a used route, albeit by an extremely small minority of travelers. However,  it is fair to conclude that these two dynamics are much less likely to satisfy RC than the other three.

Plots (b) and (d) deal with Proportionality Conservation (PC). From (d), we can see that CULO perfectly conformed to PC as predicted by the theory.  The projection dynamic and the Smith dynamic failed to conserve proportionality, as they both severely violated the first- and second-order proportionality conditions. Of the two, the projection dynamic performed worse.  The best-response dynamic outperformed the above two, although its deviation from the two proportionality conditions is still substantial.   The solution obtained by the replicator dynamic does not exactly satisfy the two proportionality conditions, but the violations are barely detectable from the plot. This behavior is expected, given the continuous version of the replicator dynamic is closely related to CULO.
From plot (b), we observe that both CULO and the replicator dynamic are capable of approaching the benchmark entropy value (the precise entropy value associated with the MEUE strategy).  All other three models achieve an entropy value markedly lower than the benchmark: the worst is the projection dynamic, followed by the best response and the Smith dynamic.

\subsubsection{Summary}

From what we saw in this section, it is safe to conclude that none of the four models discussed in Section \ref{sec:other} satisfies all of the four properties, even though they are globally stable (GS) under the assumptions adopted in this study. Specifically, the evidence strongly suggests that the best-response dynamic violates TS and PC, the projection dynamic violates RC and PC, the replicator dynamic violates TS, and the Smith dynamic violates all three.  

Two dynamics are worth a final remark.  First, like CULO, the replicator dynamic can be used to solve the MEUE problem approximately. However, the quality of the approximation degrades as the step size increases. This is a computational disadvantage because small step sizes lead to slow convergence. CULO does not suffer from this disability thanks to a superior convergence guarantee. Second, it is somewhat surprising to see the best-response dynamic, despite the poor convergence performance, can obtain a solution more closely resembling the MEUE strategy than the projection and the Smith dynamics. This empirical finding appears to confirm the conjecture put forth by \citet{florian2014uniqueness}, who argued the Frank-Wolfe algorithm can yield UE solutions that approximately obey the condition of proportionality.

\section{Conclusions}
\label{sec:conclusion}

The lack of a unique user equilibrium (UE) route flow in traffic assignment has posed a significant challenge to many transportation applications. A common remedy to this long-standing problem is the maximum entropy principle, which advocates consistently choosing the most likely UE route flow as the representative of the countless candidates.  This study provided a new behavioral underpinning for this principle. Our theory is built on a recently proposed day-to-day (DTD) dynamical model called cumulative logit, or CULO, which can reach a UE state without presuming travelers are perfectly rational.  We proved that CULO always selects (or converges to) the maximum entropy UE (MEUE) route flow given a proper initial condition. We further identified two such conditions: (i) travelers have zero prior information about routes and thus are forced to give all routes an equal choice probability, and (ii) all travelers gather information from the same source such that the so-called general proportionality condition is satisfied. Thus, the MEUE route flow may result from a routing game in which boundedly rational travelers continuously learn about and refine their valuation of the routes and adjust their routing strategy accordingly. 
The revelation suggests that CULO may be used as a solution algorithm for the MEUE route flow problem. To operationalize this idea, we proposed to bypass the route enumeration required in the original CULO model through an iterative route discovery scheme. We devised two schemes. The first guarantees convergence to UE but not MEUE. The second strives not to miss any UE route, a prerequisite for maximizing entropy. Though no theoretical assurance was provided, initial numerical results confirmed the effectiveness of the heuristic.

Having demonstrated the capability of CULO in solving the MEUE problems, we turned to address a natural question: do the other DTD models known to converge to a UE solution have a similar capability?  To answer this question, we first established the four properties underlying CULO's success, namely (i) global stability (GS), (ii) trajectory stability (TS), (iii) route conservation (RC), and (iv) proportionality conversation (PC). 
Of the four popular DTD models we examined, the replicator dynamic is the only one that has the potential to attain the MEUE solution with some regularity. However, the replicator dynamic satisfies PC  approximately only when it is discretized with a very small step size, which tends to slow the overall convergence.  The convergence of the best-response dynamic is the slowest and most disorderly, but it seems to adhere to the MEUE solution better than the projection dynamic and the Smith dynamic.

There are a few directions that future research can pursue. First, the current MEUE affirmation conditions are established for the standard routing game.  It would be useful to extend them to more general games, such as those with heterogeneous users and non-separable cost functions. To the best of our knowledge, few had considered the MEUE problem in these general routing games, and unlike the standard game, no specialized MEUE algorithm has ever been developed. Due to its simplicity and flexibility,  CULO can easily fill this gap if the results given by this paper can be generalized.  Another interesting question is whether we can design a route discovery scheme that can find all UE routes.  It is possible that Algorithm \ref{alg:link-based} already possesses this capability if we set the noise term properly and simply let the process run indefinitely.  Either way, a more rigorous theoretical investigation is warranted. 
{Our analysis of the continuous dynamical models left many questions unanswered.}  To name a few: why does the projection dynamic appear to satisfy TS? Can the limiting point of the discrete version of the replicator dynamic always make a close approximation of MEUE? If so, under what conditions? How do we explain the vastly different behavior between the Smith dynamic and the replicator dynamic, given they resemble each other so strikingly? 
{Finally, MEUE bears intriguing similarities with some network design problems, especially the entropy-based estimation of origin-destination (O-D) matrix \citep[e.g.,][]{van1980most}, in that they all involve selecting an equilibrium to optimize an entropy function.  By this analogy, the initial state in our model plays the role of the prior (or historical) matrix in O-D estimation. A future study may exploit this connection for the purpose of solving certain network design problems through a DTD dynamical process.
}

\section*{Acknowledgements}

This research is funded by the US National Science Foundation's Civil Infrastructure System (CIS) Program under the award  CMMI \#2225087.  The authors are grateful for the valuable comments offered by Prof. Yafeng Yin at the University of Michigan, Ann Arbor, and Prof. Zhaoran Wang at Northwestern University. The remaining errors are our own.

\bibliographystyle{ims}
\begin{small}
\bibliography{example_paper}
\end{small}

\newpage
\begin{appendices}

\section{Explanation of the entropy function}
\label{app:entropy}
{
In Section 2.1, we have defined the negative entropy of a route choice strategy $\vp \in \sP$ as
    \begin{equation*}
        \phi(\vp) = \langle \diag(\vq) \vp, \log(\vp)\rangle.
    \end{equation*}
    To explain how this measures ``the number of states" (i.e., the different ways travelers can be arranged to produce the route flow corresponding to $\vp$), suppose that the flow carried by a single traveler is $\tau > 0$ ($\tau$ is a small constant). Hence, the number of travelers traveling between each OD pair $w \in \sW$ and selecting each route $k \in \sK_w$ would be $n_w = q_w / \tau$ and $m_k = f_k / \tau$, respectively. 
    Applying the basic counting principle, the total number of states, after taking the logarithm, reads
    \begin{align*}
        &\log\big(\prod_{w \in \sW} \frac{n_w!}{\prod_{k \in \sK_w} m_k!}\big) = \sum_{w \in \sW} \big(\log(n_w!) - \sum_{k \in \sK_w} \log(m_k!)\big) \\
        &\qquad =\sum_{w \in \sW} \big(n_w \log(n_w) - n_w + \mathcal{O}(\log(n_w)) - \sum_{k \in \sK_w} m_k \log(m_k) - m_k + \mathcal{O}(\log(m_k))\big), \\
        &\qquad  = \sum_{w \in \sW} \big(n_w \log(n_w) + \mathcal{O}(\log(n_w)) - \sum_{k \in \sK_w} m_k \log(m_k) + \mathcal{O}(\log(m_k))\big),
    \end{align*}
    where Stirling's formula gives the second equality, while the relation $n_w  = \sum_{k \in \sK_w} m_k$ gives the third one. When $\tau$ is sufficiently small (as close to the nonatomic setting), we would have $n_w \to \infty$ and $m_k \to \infty$. Hence, the term $\mathcal{O}(\log(n_w))$ and $\mathcal{O}(\log(m_k))$ would become eligible relative to $n_w \log(n_w)$ and $m_k \log(m_k)$. Further noting that
    \begin{align*}
        \sum_{w \in \sW} \big( n_w \log(n_w) - \sum_{k \in \sK_w} m_k \log(m_k) \big) &= \sum_{w \in \sW}  \sum_{k \in \sK_w} m_k  \big(\log(n_w) - \log(m_k) \big) =- \frac{1}{\tau} \cdot \sum_{w \in \sW}  q_w \sum_{k \in \sK_w} p_k \log(p_k),
    \end{align*}
    we have
    \begin{align*}
        \log\big(\prod_{w \in \sW} \frac{n_w!}{\prod_{k \in \sK_w} m_k!}\big) \approx -\frac{1}{\tau} \cdot \sum_{w \in \sW}  q_w \sum_{k \in \sK_w} p_k \log(p_k) = -\frac{1}{\tau} \cdot \langle \diag(\vq) \vp, \log(\vp) \rangle,
    \end{align*}
    when $\tau$ is sufficiently small.  Dropping $\tau$ (which is a constant) then gives rise to the entropy function for evaluating the likelihood of $\vp$ in our setting. 

}
\section{Proofs in Section \ref{sec:main}}

\subsection{Proof of Lemma \ref{lm:proportionality}}
\label{app:proportionality}
\begin{proof}
    We first note that for all $\vs^t \in \sR^{|\sK|}$, the corresponding logit choice can be written in the vector form as $q_r(\vs^t) = \vy^t / \mSigma^{\T} \mSigma \vy^t$, where $\vy^t = \exp(-r \cdot \vs^t)$.  As each column of $\mSigma$ is a standard unit vector, we have $\log(\mSigma^{\T} \mSigma \vy^t) = \mSigma^{\T} \log(\mSigma \vy^t)$. Denoting $\vx^t = \bar \mLambda \vp^t$ for all $t \geq 0$, we can show
    \begin{equation}
    \begin{split}
        \langle \ve, \log(\vp^t) \rangle &= \langle \ve, \log(\vy^t) - \log(\mSigma^{\T} \mSigma \vy^t) \rangle = \langle \ve, -r \cdot \vs^t - \mSigma^{\T} \log(\mSigma \vy^t) \rangle \\
        &= \langle \ve, -r \cdot \vs^t \rangle 
        = -r \cdot \sum_{i = 0}^{t - 1} \eta^i \cdot \langle \ve, \mLambda^{\T} \vx^i\rangle + \langle \ve, -r \cdot \vs^0 \rangle = -r \cdot \langle \ve, \vs^0 \rangle.
    \end{split}
    \end{equation}
    The first and the second equalities hold due to the earlier discussions, the third and fifth equalities hold because $\ve \in \ker{(\mSigma)}$ and $\ve \in \ker{(\mLambda)}$, respectively, and the fourth equality is obtained by applying Equation \eqref{eq:culo}.
\end{proof}

\subsection{Proof of Lemma \ref{lm:kl-condition}}
\label{app:kl-condition}

\begin{proof}
    Per Proposition \ref{prop:solution-set-2}, the KL projection problem \eqref{eq:main} can be  written as
    \begin{equation}
        \begin{split}
            \min_{\vp^* \geq 0}~~& \langle \diag(\vq) \vp^*, \log(\vp^*) - \log(\vp^0) \rangle.\\
            \text{s.t.}~~&\mSigma \diag(\vq) \vp^* = \vd, \quad \mLambda \diag(\vq)\vp^* = \vx^*,
            \label{eq:KL-minimization}
        \end{split}
    \end{equation}
    which is evidently a convex program. Hence, any $\bar \vp^* \in \sP^*$ solves \eqref{eq:KL-minimization} if and only if there exist $\valpha \in \sR^{|\sK|}$ and $\vbeta \in \sR^{|\sA|}$ such that 
    \begin{equation}
        \begin{cases}
        \bar \vp^* \geq 0, \quad \log(\bar \vp^*) - \log(\vp^0) - \mSigma^{\T} \valpha - \mLambda^{\T} \vbeta  \geq 0, \\
        \langle \bar \vp^*, \log(\bar \vp^*) - \log(\vp^0) - \mSigma^{\T} \valpha - \mLambda^{\T} \vbeta \rangle = 0.
        \end{cases}
        \label{eq:kkt}
    \end{equation}
    If $\langle \ve, \log(\bar \vp^*) - \log(\vp^0) \rangle = 0    $
    for all $\ve \in \ker{(\mSigma)} \cap \ker{(\mLambda)}$, then  $\log(\bar \vp^*) - \log(\vp^0) \in \im(\mSigma^{\T}) \cup \im(\mLambda^{\T})$\footnote{For any matrix $\mA \in \sR^{m \times n}$, we have $\ker(\mA)^{\perp} = \im(\mA^{\T})$, i.e., the perpendicular complement of $\ker(\mA)$ is $\im(\mA^{\T})$.}, which means one can always find $\valpha \in \sR^{|\sK|}$ and $\vbeta \in \sR^{|\sA|}$ such that 
    \begin{align*}
        \log(\bar \vp^*) - \log(\vp^0) =   \mSigma^{\T} \valpha + \mLambda^{\T} \vbeta. \end{align*}
Thus, Condition \eqref{eq:kkt} must be satisfied.  
\end{proof}

\subsection{Proof of Theorem \ref{thm:main}}
\label{app:main}

\begin{proof}
    Per Lemma \ref{lm:proportionality}, for all $\ve \in \ker{(\mSigma)} \cap \ker{(\mLambda)}$, the value of $\langle \ve, \log(\vp^t) \rangle$ is the same for all $t \geq 0$ given $\vp^0$. This observation leads to
    \begin{equation}
        \langle \ve, \log(\vp^t) - \log(\vp^0) \rangle = 0.
        \label{eq:thm-1}
    \end{equation}
    Noting that the function $\langle \ve, \log(\vp) - \log(\vp^0)\rangle$ is continuous in $\vp$, we obtain $\langle \ve, \log(\bar \vp^*) - \log(\vp^0) \rangle = 0$ by letting $t \to \infty$ in Equation \eqref{eq:thm-1}.  This implies $\bar \vp^*$ is the KL projection of $\vp^0$ on $\sP^*$ per Lemma \ref{lm:kl-condition}.
\end{proof}

\subsection{Proof of Corollary \ref{cor:no-left}}
\label{app:no-left}

\begin{proof}
    We first note that $\supp(\bar \vp^*) \subseteq \sK^*$ because $\bar \vp^* \in \sP^*$. Suppose $\bar \vp^*$ is the limiting point of the CULO model but there exists $k \in \sK^*$ such that $k\notin \supp(\bar \vp^*)$, then we must have $\bar \evp_k^* = 0$.  Construct a function $l(\epsilon) = D((1 - \epsilon) \cdot \bar \vp^* + \epsilon \cdot \vp^0, \vp^0)$, where $D(\cdot,\cdot)$ is the KL divergence defined in \eqref{eq:kl-def}. The reader can verify that the derivative of $l(\cdot)$ at $\epsilon=0$ is $-\infty$.  This means moving from $\bar \vp^*$ toward $\vp^0$ can reduce the KL divergence. Hence, $\bar \vp^*$ cannot be the solution to the KL projection problem \eqref{eq:main}, or the limiting point of the CULO model, a contradiction. 
\end{proof}

\subsection{Proof of Proposition \ref{prop:main-1}}
\label{app:main-1}

\begin{proof}
    If $\vs^0 = \mLambda^{\T} \vv^0$, then for any $\ve \in \ker{(\mSigma)} \cap \ker{(\mLambda)}$, we have $\langle \ve, \vs^0 \rangle = \langle \ve, \mLambda^{\T} \vv^0 \rangle = \langle \mLambda \ve, \vv^0 \rangle = 0$. Thus, Lemma \ref{lm:proportionality}  guarantees $\langle \ve, \log(\vp^t) \rangle = -r \cdot  \langle \ve, \vs^0 \rangle =0$ for any $\ve$ and $t \geq 0$.  As the function $\langle \ve, \log(\vp)\rangle$ is continuous in $\vp$, letting $t \to \infty$  leads to $\langle \ve, \log(\bar \vp^*)\rangle = 0$ for any $\ve$, which implies $\bar \vp^*$ is the MEUE strategy according to Proposition \ref{prop:proportionality}.
\end{proof}

\section{Proofs in Section \ref{sec:generation}}

\subsection{Proof of Proposition \ref{prop:alg-1}}
\label{app:alg-1}

\begin{proof}
    First, as there are a finite number of acyclic paths in a network, the discovery process must stop adding new routes after finite days (note that cyclic paths can never be a shortest route as long as the link cost is strictly positive). That is, there must exist $T_1 < \infty$ and $\overline \sK_+ \subseteq \sK$ such that $\sK_+^t = \overline \sK_+$ for all $t \geq T_1$. Starting from $t = T_1$, Algorithm \ref{alg:route-based} reduces to the original CULO model without route exploration, applied to solving a ``reduced" routing game in which only routes in $\overline \sK_+$ are available. Denote the route-link incidence matrix corresponding to $\overline \sK_+$ as $\overline \mLambda_+^t$ and define  $c_+: \overline \sP_+ \to \sR^{|\overline \sK_+|}$ as a map that satisfies
\begin{equation}
    c_+(\vp_+) =  \overline \mLambda_+^{\T} u(\vx), \quad \text{where}~\vx = \overline \mLambda_+ \diag(\overline \mSigma_+^{\T} \vd) \vp_+.
    \label{eq:c_plus}
\end{equation}
By Proposition \ref{thm:convergence-ue}, as long as $r < 1/2\overline L$ for some $L \geq \max_{\vp_+ \in \overline \sP_+} \|\nabla c_+(\vp_+)\|_2$, the route choice strategy $\vp_+^t$ must converge to a fixed point $\overline{\vp}_+ \in \overline \sP_+$, which is a UE of the reduced routing game.

We then claim $\overline \vp = [\overline \vp_+; \vzero] \in \sP^*$. To simplify the proof, let us assume $|\sW| = 1$ without loss of generality; otherwise, we can simply pick one $w \in \sW$ to raise the following conflict. Suppose that $\overline \vp \notin \sP^*$, then given any $k_0 \in \argmin_{k \in \sK} c_k(\overline \vp)$, we must have $c_0 := c_k(\overline \vp) < c_{\min} := \min_{k \in \overline \sK_+} c_k(\overline \vp)$. Hence, $k_0 \in \sK \setminus \overline \sK_+$, i.e., there  exist some $k_0 \in \sK \setminus \overline \sK_+$ strictly better than any routes in $\overline \sK_+$.
For notational simplicity, let us define $\vp^t = [\vp_+^t; \vzero]$ for all $t \geq T_1$.  Since $\vp_+^t \to \overline{\vp}_+ \Rightarrow \vp^t \to \overline \vp \Rightarrow c(\vp^t) \to c(\overline \vp)~\text{and}~\min_{k \in \overline \sK_+} c(\vp^t) \to c_{\min}$, there must exist $T_2 > T_1$ such that whenever $t \geq T_2$, we have
\begin{equation}
    c_{k_0}(\vp^t) < c_0 + \delta / 3 \quad \text{and} \quad \min_{k \in \overline \sK_+} c_{k}(\vp^t) > c_{\min} - \delta / 3,
    \label{eq:conflict}
\end{equation}
where $\delta = c_{\min} - c_0$. This means on day $T_2$,  route $k_0$ is strictly better than all routes in $\overline{\sK}_+$. Hence, the route discovery process has not stabilized at $T_2$, which contradicts with the assumption that  $\sK_+^t$ remains unchanged after $t \geq T_1$.
\end{proof}

\subsection{Proof of Proposition \ref{prop:alg-2}}
\label{app:alg-2}
\begin{proof}
    Let $\sK_+^t = \overline \sK_+$ for all $t \geq T_1$ for some $T_1 > 0$. Starting from $T_1$, Algorithm \ref{alg:link-based} is reduced to the original CULO model applied to solving a routing game defined on  $\overline \sK_+$. By viewing $t = T_1$ as an initial point of Algorithm \ref{alg:link-based}, Condition (B) given in Proposition \ref{prop:main-1} is  satisfied. Hence,  $\overline \vp_+ = \lim_{t \to \infty} \vp_+^t$ must be the MEUE strategy of the reduced problem. 
    
    We proceed to prove $\overline \vp = [\overline \vp_+; \vzero]$ is the MEUE strategy of the original routing game as long as $\overline \sK_+ \supseteq \cup_{\vp^*} \supp{(\vp^*)}$. We first define $\phi_+: \overline \sP_+ \to \sR$ as the negative entropy function of the reduced routing game, which reads
    \begin{equation}
        \phi_+(\vp_+) = \langle \diag(\overline \vq_+) \vp_+, \log(\vp_+)\rangle, \quad \text{where}~\overline \vq_+ = (\overline \mSigma_+)^{\T} \vd.
    \end{equation}
    Denoting $\overline \sP_+^* \subseteq \overline \sP_+$ as the set of UE strategies for the reduced routing game, we  claim 
    \begin{equation}
        \min_{\vp_+^* \in \overline \sP_+^*} \phi_+(\vp_+^*) \leq \min_{\vp^* \in \sP^*} \phi(\vp^*).
        \label{eq:lower}
    \end{equation}
    To see this, consider a map $h: \sP^* \to \overline \sP_+$ such that $h(\vp^*) = (\evp_k^*)_{k \in \overline \sK_+}$, i.e., it ``cuts off" all elements in $\sK \setminus \overline \sK_+$. As $\cup_{\vp^*} \supp{(\vp^*)} \subseteq \overline \sK_+$, the elements dropped by the map $h$ must all be zero. Hence, we conclude that for all $\vp^* \in \sP^*$, (1) $h(\vp^*) \in \overline \sP_+^*$; (2) $\phi_+(h(\vp^*)) = \phi(\vp^*)$. Combining both, Equation \eqref{eq:lower} must hold.  
    Recalling that $\phi_+(\overline \vp_+) =  \min_{\vp_+^* \in \overline \sP_+^*} \phi_+(\vp_+^*)$, we derive $\phi_+(\overline \vp_+) \leq \min_{\vp^* \in \sP^*} \phi(\vp^*)$. Finally,  as $\overline \vp = [\overline \vp_+; \vzero] \in \sP^*$ and $\phi_+(\overline \vp_+) = \phi(\overline \vp)$,  we have $\phi(\overline \vp) = \min_{\vp^* \in \sP^*} \phi(\vp^*)$, which means $\overline \vp$ is the MEUE of the original routing game. 
\end{proof}

\end{appendices}

\end{document}